\documentclass[11pt]{article}

\usepackage{amssymb}
\usepackage{algorithmicx}
\usepackage[ruled]{algorithm}
\usepackage{algpseudocode}
\usepackage{algpascal}
\usepackage{algc}

\alglanguage{pseudocode}
\usepackage{graphicx}
\usepackage{amsmath,epsfig,graphics}

\newenvironment{proof}[1][Proof]{\textbf{#1.} }{\ \rule{0.5em}{0.5em}}

\def\N{{\rm I\!N}}
\newtheorem{theorem}{Theorem}[section]

\newtheorem{proposition}[theorem]{Proposition}

\newcommand{\Frac}[2] {\frac{\textstyle #1} {\textstyle #2}}

\begin{document}

\title{A simple mathematical approach to optimize the structure of reaction-diffusion physicochemical systems}
\author{J.-P. Chehab\thanks{
Laboratoire Amienois de Math\'ematiques Fondamentales et Appliqu\'ees (LAMFA), {\small UMR} 7352,
 Universit\'e de Picardie Jules Verne, 33 rue Saint Leu, 80039 Amiens France
 , ({\tt
 Jean-Paul.Chehab@u-picardie.fr}).}, A. A. Franco\thanks{Laboratoire de R\'eactivit\'e et de Chimie des Solides (LRCS)
Universit\'e  de Picardie Jules Verne - CNRS / UMR 7314, 33, rue St. Leu, Amiens, France F-80039({\tt
 Alejandro.Franco@u-picardie.fr}) and
R\' eseau sur le Stockage Electrochimique de l'Energie (RS2E),
FR CNRS, 3459,
France}, 
Y. Mammeri\thanks{Laboratoire Amienois de Math\'ematiques Fondamentales et Appliqu\'ees  (LAMFA), {\small UMR} 7352,
 Universit\'e de Picardie Jules Verne, 33 rue Saint Leu, 80039 Amiens France  ({\tt
 Youcef.Mammeri@u-picardie.fr}).}}

\date{ }

\maketitle

\begin{abstract}
The calculation of optimal structures in reaction-diffusion models is of great importance in many physicochemical systems. We propose here a simple method to monitor the number of interphases for long times by using a boundary flux condition as a control.  We consider as an illustration a 1-D Allen-Cahn equation with Neumann boundary conditions. Numerical examples are given and perspectives for the application of this approach to electrochemical systems are discussed.
\end{abstract}
\section{Introduction}
The dynamics of a large diversity of physicochemcal systems can be mathematically modeled as reaction-diffusion systems in which it is described  how the composition of multiple chemical species distributed in space change under the influence of competitive chemical reactions between the species (giving origin to a new species)  and the diffusion which causes the species to spread out  in the space.
It is well known that depending on the relative importance of the kinetics and the diffusion these systems can provide a large diversity of behaviors, including the formation of  complex structures and patterns see \cite{Sachs}.

Such a structure formation occurs for example during the solid phase formation and evolution in intercalation and conversion reactions in rechargeable lithium batteries \cite{Franco1,Franco2}, during the self-organisation of materials  occuring with the fabrication process of composite electrodes for electrochemical devices applications \cite{Malek1}, during the microstructural evolution of composite elecrodes upon their degradation \cite{Malek2}  and in other competitive chemically reactive systems like in the Belousov-Zhabotinsky reaction \cite{Sirimungkala}.\\

Designing appropriate controllers of these reaction-diffusion systems can reveal of great relevance within a reverse engineering approach for example towards the optimization of discharge-charge of lithium batteries (by for example enhancing the formation of solid phases during discharge more reversible upon charge) and the optimization of  the structure of  the fabricated electrodes as function of the fabrication parameters (e.g. temperature dynamics, reactant flow, etc.).\\

In this paper, we consider the one-dimensional Allen-Cahn equation 
\begin{eqnarray}
\Frac{\partial u}{\partial t} -\Frac{\partial^2 u}{\partial x^2} +\Frac{1}{\epsilon^2}f(u)=0 & x\in]0,1[, t >0,\\
u_x(0,t)=\alpha(t), u_x(1,t)=0& \forall t >0,\\
u(x,0)=u_0(x) & x\in ]0,1[.
\end{eqnarray}
\\
This reaction-diffusion equation describes the process of phase separation in  many situations.
It was originally introduced in \cite{AllenCahn} by Allen and Cahn to model the
motion of anti-phase boundaries in crystalline solids. In equation (1), $u$ represents the concentration
of one of the possible phases, $\epsilon$ represents the interfacial
width, supposed to be small as compared to the characteristic length of the laboratory scale. The homogenous
Neumann boundary condition (when $\alpha(t)=0$) traduces  that there is no loss of mass across the boundary walls.
However, the Allen-Cahn equation is invoqued in a large number of complicated
moving interface problems in materials science through a phase-field
approach, therefore a  large litterature in mathematical analysis 
and in numerical analysis is devoted to the study of the mathematical properties of this equation and of its simulation (see \cite{MPierre,JShen} and the references therein).\\

In equation (1), $f(u)$ represents the potential energy and $\alpha(t)$ represents the control flux at one of the boundaries; $f(u)$ is assumed having
stable roots $\rho_i$, $i=1,\cdots, r$ such that $f(\rho_i)=0$ and $f'(\rho_i)>0$.
It is observed in many cases that when $\epsilon <<1$ and as $t$ goes to $+ \infty$ , the solutions tend to steady states ${\bar u}$ which consist in (almost) piecewise constant functions whose the different values are equal to the stable roots of $f$ which represent the different phase stripes. Hence ${\bar u}$ exhibits large gradient near 
$\rho_i$, as illustrated in Figure (\ref{fig1}).
\begin{figure}[h]
\begin{center}
\includegraphics[height=7.0cm,width=8.5cm]{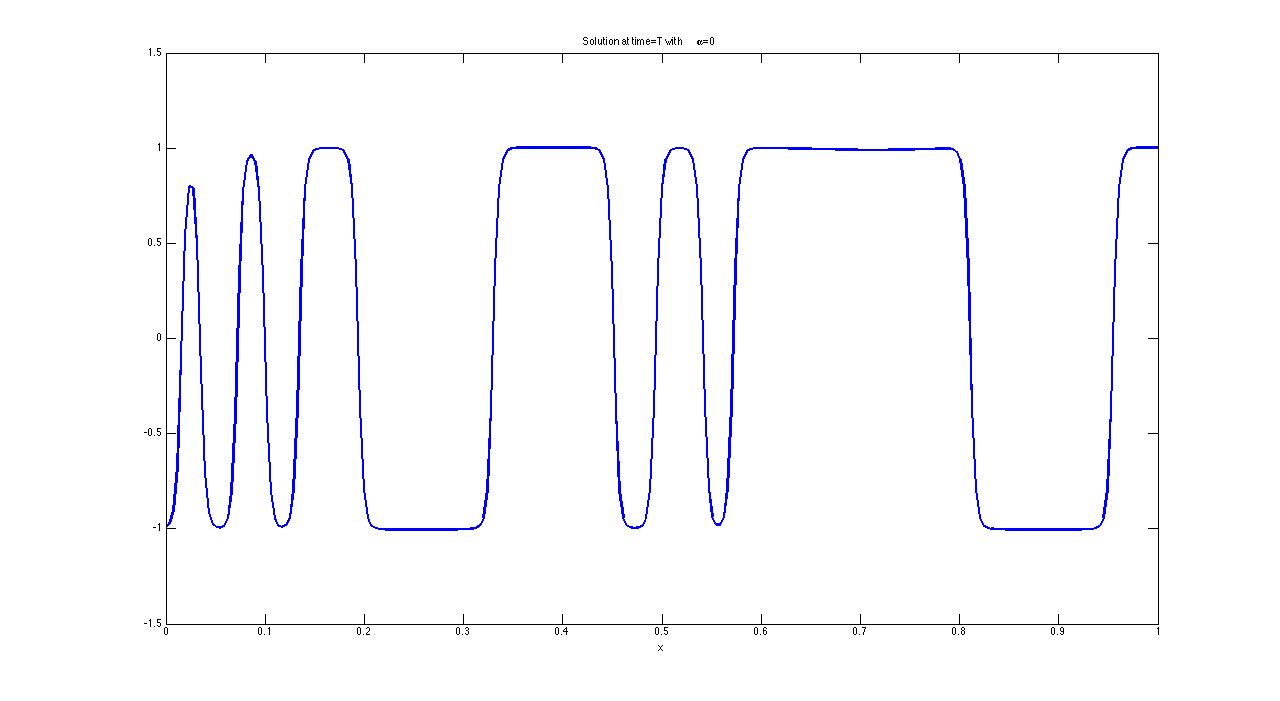}
\includegraphics[height=7.0cm,width=8.5cm]{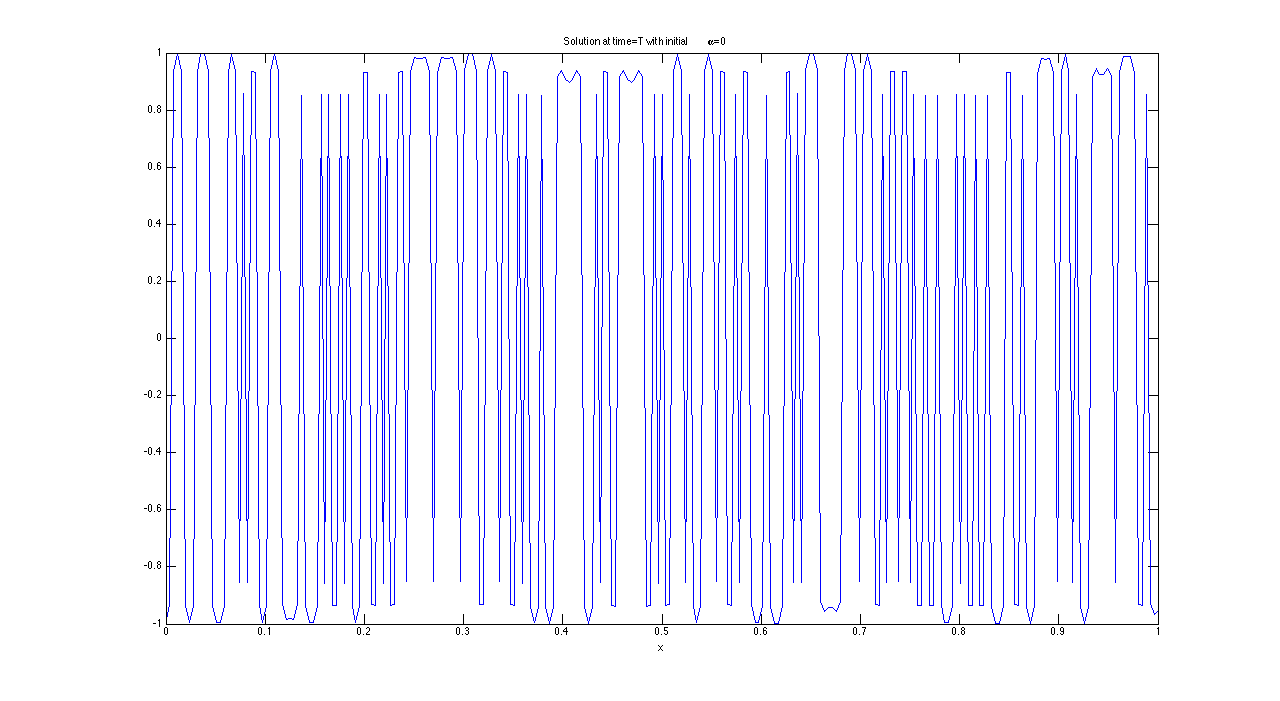}
\end{center}
\caption{Steady state for $\epsilon=0.004$ (left) and for $\epsilon=0.001$ (right).}
\label{fig1}
\end{figure}
\\
\\
An important issue in the conception of rechargeable lithium and post-lithium batteries, is the design of active materials providing upon the battery discharge a number of interphases as low as possible. The morphological simplicity of such discharged materials is expected to enhance the rechargeability of these type batteries and thus to increase their efficiency \cite{Franco2}.
In this paper we propose a first numerical strategy to calculate the boundary flux function $\alpha(t)$ on a given time interval $[0,T]$, with $T$ large enough, in such a way the number of interphase of the steady state ${\bar u}$ is minimized. To this end, we  consider as  control function $\alpha(t)$, $\epsilon$ being constant.

For the sake of simplicity, we first restrict ourselves to the case $f(u)=u(u^2-1)$ which possesses 3 roots: $u=\pm 1$ which are stable and $u=0$ which is unstable.\\

The article is organized as follows: first, in Section 2, we present first the global numerical strategy by deriving the estimation of the number of interphases, which will be the merit function to minimize. Then, we present the finite differences discretization of the system in space and we describe the numerical solver, that includes the optimization process as well as the full discretized problem to be solved at each iteration. In Section 3, we present some numerical results demonstrating  the numerical controllability of the problem: we calculate optimal$\alpha$ for different values of $u_0$, $T$ and $\epsilon$. Finally, in Section 4 we conclude and indicate further perspectives of development of our work.
\section{Numerical strategy}
\subsection{Estimation of the number of interphases}
We consider the finite differences discretization in space of the Allen-Cahn equation which leads to a differential system. The grid points $x_i$, $i=1,\cdots, N$ are regularly spaced for simplicity, $h$ is the corresponding stepsize.
We assume that $h$ is small enough in order the discrete solution captures the strong gradients near the interphases.

The steady solution ${\bar u}$ is considered to be almost piecewise constant, so its approximations at grid points ${\bar u}_i$,
$i=1,\cdots, N$ take the values $\pm 1$. Hence
$$
{\bar u}_{i+1}-{\bar u}_i 
=\left\{
\begin{array}{c}
0 \\
2 \\
-2  \\
\end{array}
\right.
$$
Therefore, the number of interphases is
\begin{eqnarray}
\label{formula_changes}
N({\bar u})&=\Frac{1}{2}\displaystyle{\sum_{i=0}^{N}\mid {\bar u}_{i+1}-{\bar u}_i  \mid}.
\end{eqnarray}
This quantity can be related to the $L^1$-norm of $u'$, indeed
\begin{eqnarray}
N({\bar u})=\Frac{1}{2}\displaystyle{\sum_{i=0}^{N}\mid \Frac{{\bar u}_{i+1}-{\bar u}_i}{h}  \mid h}\simeq
\Frac{1}{2}\displaystyle{\int_0^1 \mid {\bar u}'(x)\mid dx}.
\end{eqnarray}
In Figure (\ref{fig1}) (left), we count 10 changes, the result given by formula (\ref{formula_changes}) is 9.9968 and
in Figure (\ref{fig1}) (right) 48 changes are counted while (\ref{formula_changes}) estimation is 47.7475.\\

We remark that an interesting numerical issue could be to plug an adaptive grid strategy since the steady solution needs only few points to be represented.
\subsection{Selection of given phases}
Our approach applies when more than 2 interphases are present. Indeed, consider for simplicity the case of $m$ stable phases. To obtain the number of interphases, it is sufficient to split the final signal profile into 4 parts, each of them reprensenting the state of one phase stirp  (see figure below in the case of $m=4$). Once done, we can apply  formula (\ref{formula_changes}) separately. This procedure allows using a weighted merit function
\begin{eqnarray}
F(u)=\displaystyle{\sum_{i=1}^m\omega_iN_i(u)},
\end{eqnarray}
where $N_i(u)$ is the number of connex components for phase $i$ and $\omega_i \ge 0$ the associated weight:  a large value of $\omega_i$ enforces the optimal state to provide small number of phases of type $i$. So it is possible to select a given profile. It has to be pointed out that $\displaystyle{\sum_{i=1}^mN_i(u)}\neq N(u)$, the total number of interphases, in fact $\displaystyle{\sum_{i=1}^mN_i(u)}\ge N(u)$, however the simplified summation formula allows selecting easily given phase stripes (Figure (\ref{fig3})).
\begin{figure}[h]
\begin{center}
\includegraphics[height=7.0cm,width=9cm]{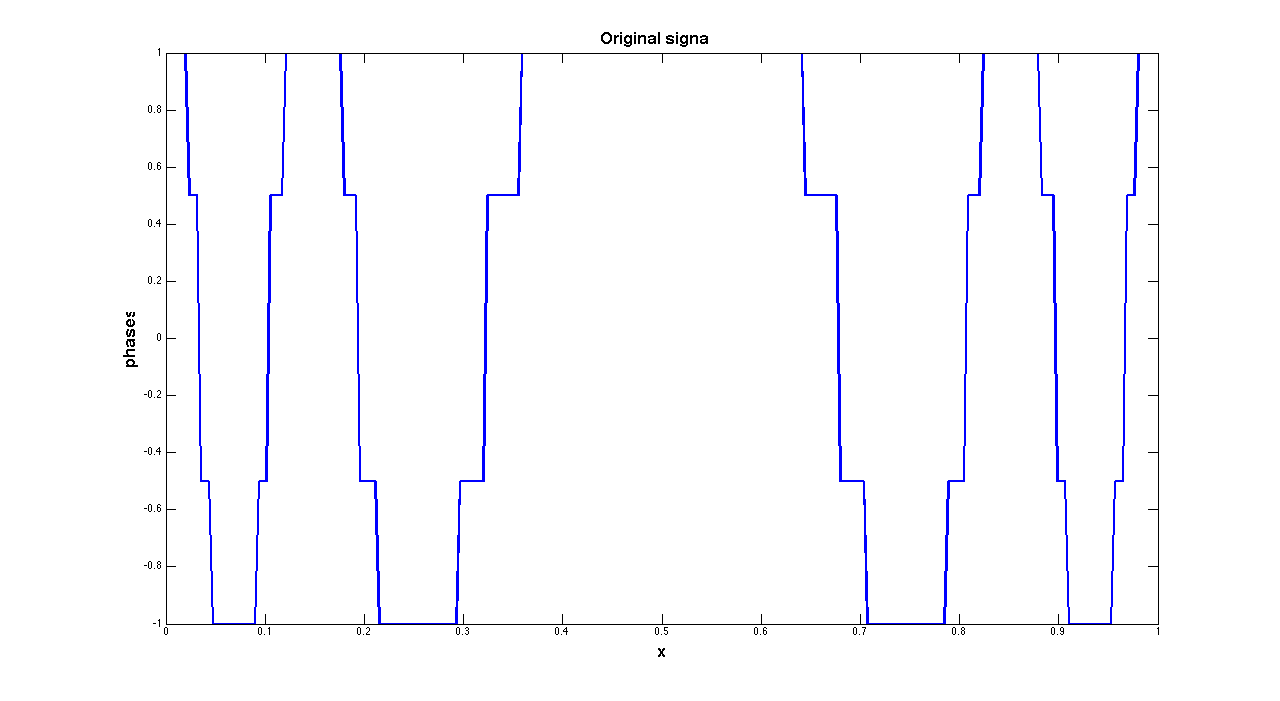}
\end{center}
\caption{The original signal}
\label{fig2}
\end{figure}
\clearpage
\begin{figure}[h]
\begin{center}
\includegraphics[height=7.0cm,width=8cm]{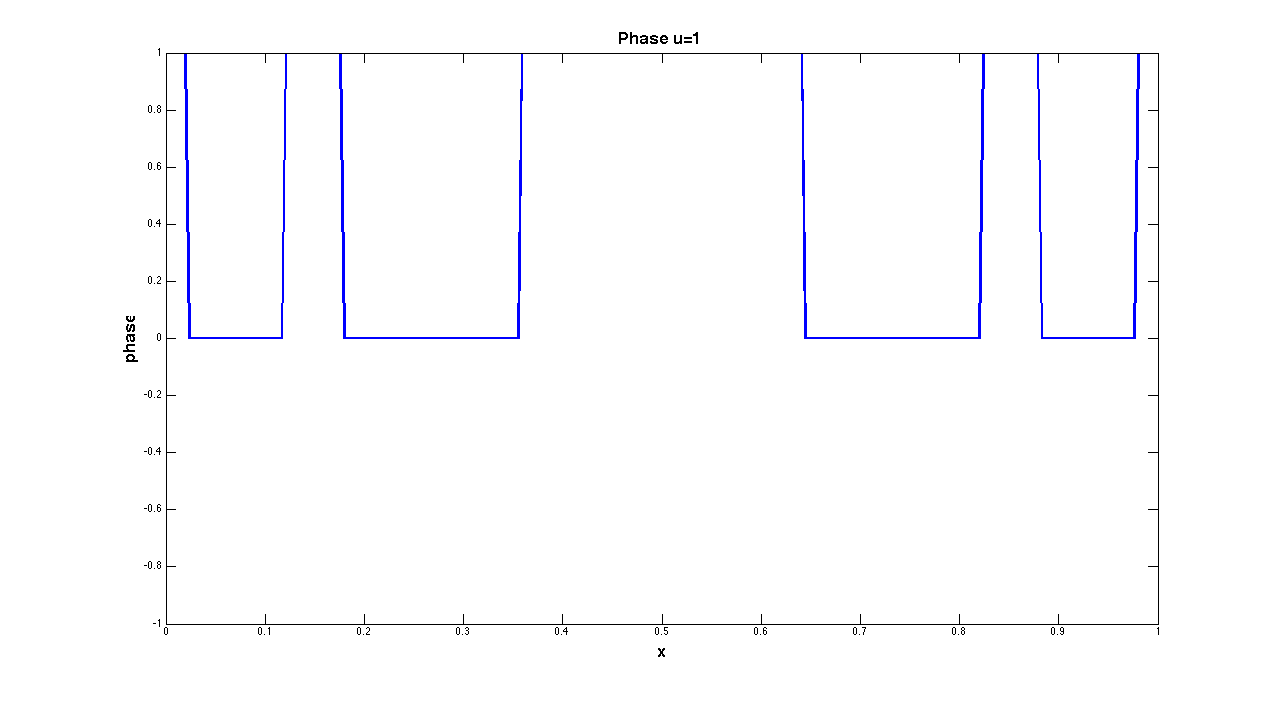}
\hskip 0.2cm
\includegraphics[height=7.0cm,width=8cm]{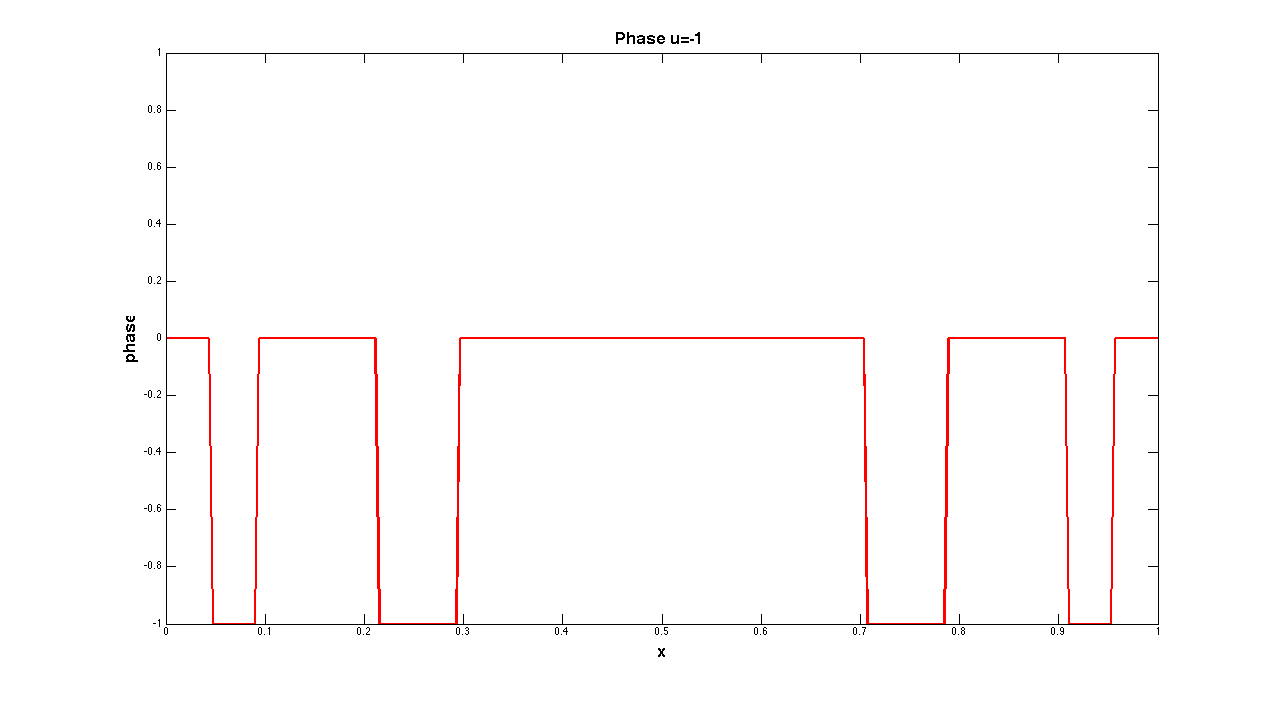}
\includegraphics[height=7.0cm,width=8cm]{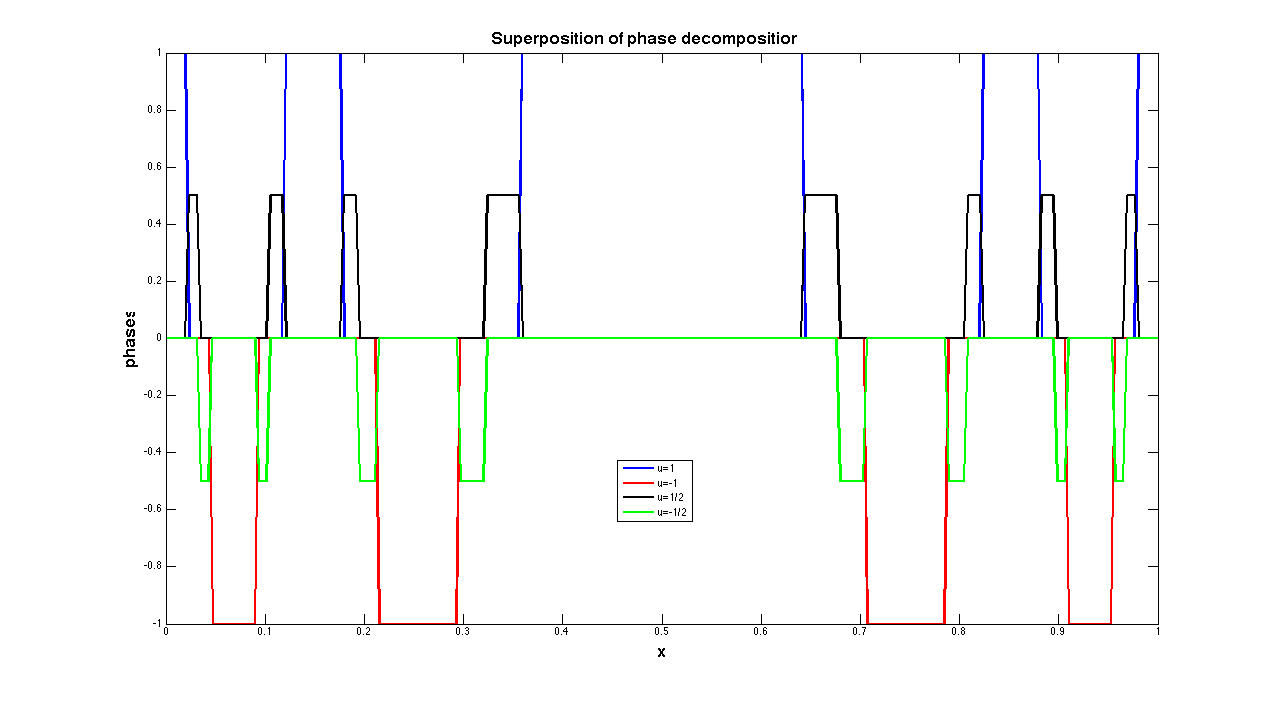}\\
\includegraphics[height=7.0cm,width=8cm]{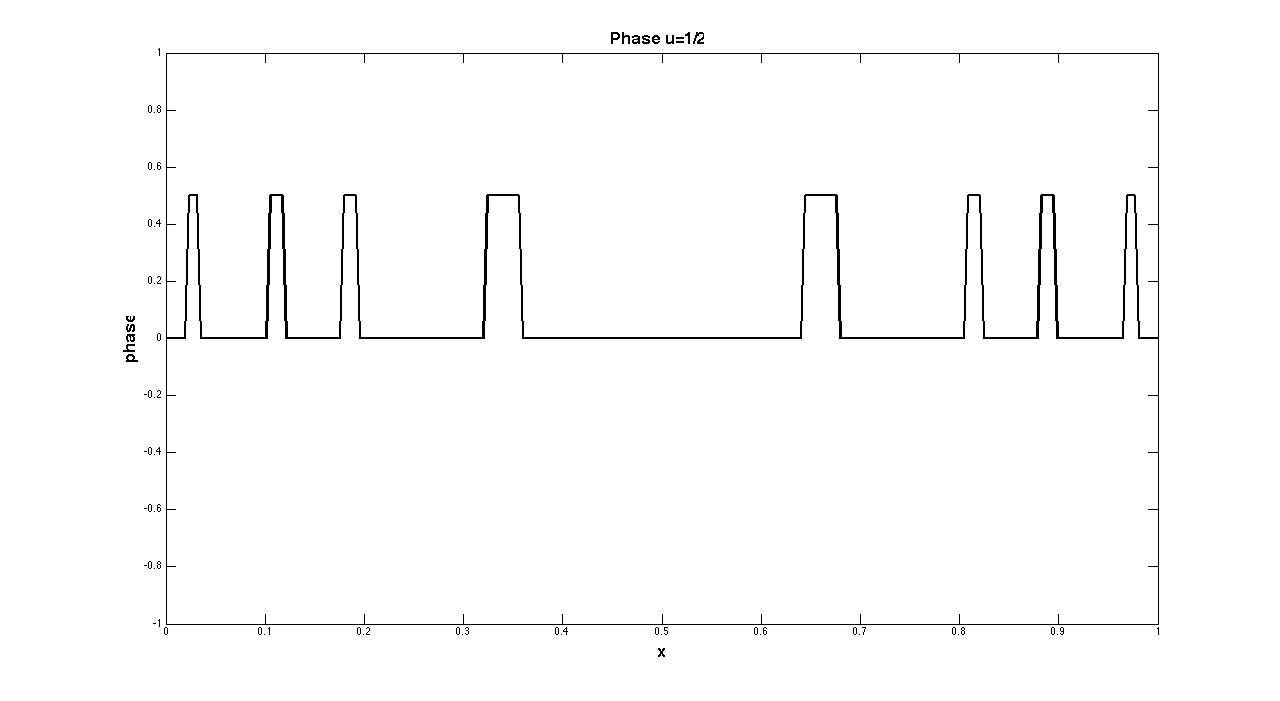}
\hskip 0.2cm
\includegraphics[height=7.0cm,width=8cm]{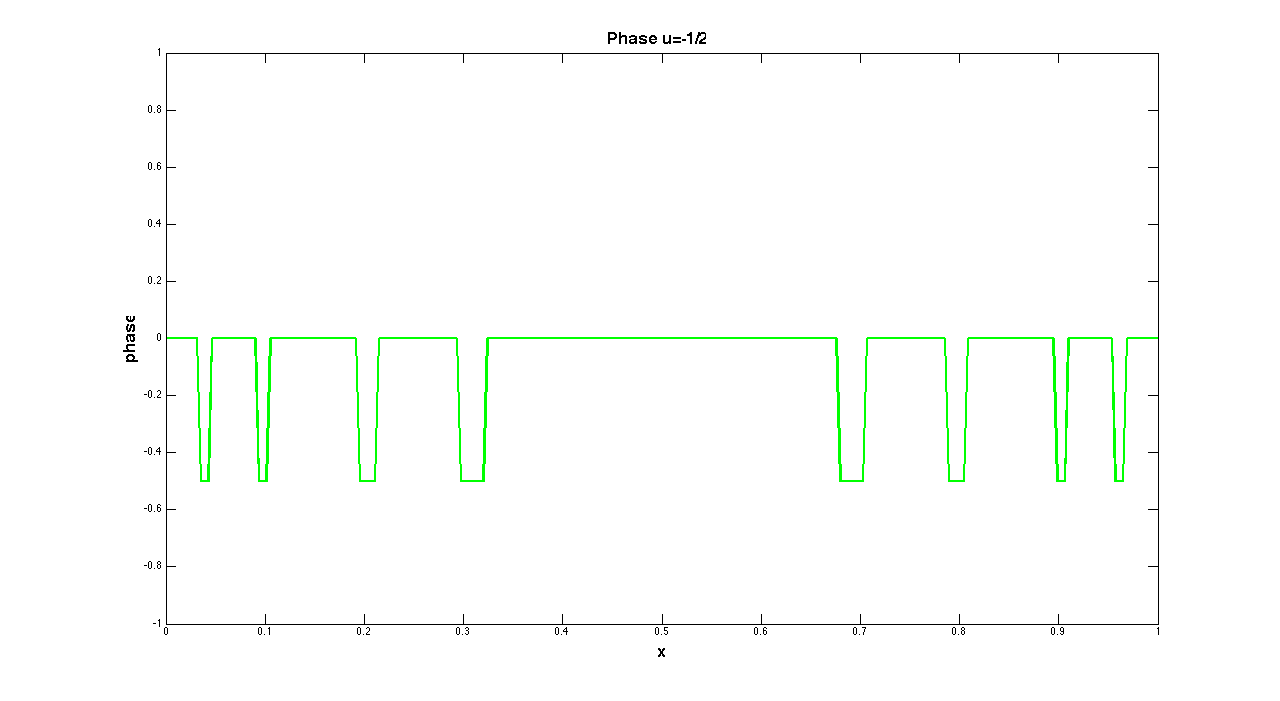}
\end{center}
\caption{Multiphase decomposition of a signal, superposition of the various phases is represented at the center}
\label{fig3}
\end{figure}
\clearpage
When the different phases $U_I, i=1,\cdots, m$ are known {\it  a priori}, the merit function can be defined more precisely as
\begin{eqnarray}
\label{weighted}
F(u)=\displaystyle{\sum_{i=1}^m\omega_i\parallel u-U_i\parallel^2}.
\end{eqnarray}
For instance, with $m=2$ and $f(u)=u(u^2-1)$, the two stable phases are $U_{1,2}=\pm 1$ and the merit function associated to $U_1=+1$ is
$$
F(u)=\parallel u-1\parallel^2,
$$
which can be considered  directly in the continuous case with, e.g., $L^2$ norm, giving rise to the merit function 
$F(u)=\displaystyle{\int_0^1 (u(x)-1)^2dx}$.
\subsection{Global scheme}
We denote by $u_i(t)$ the approximation of $u(x_i,t)$ generated by the semi discrete scheme
\begin{eqnarray}
\label{eq1}
\mbox{for } i=1,\cdots, N \ \Frac{du_i(t)}{dt} +\Frac{2u_i(t)-u_{i-1}(t)-u_{i+1}(t)}{h^2}+\Frac{1}{\epsilon^2}u_i(t)\left(u_i^2(t)-1\right)=0& t>0,\\
\label{eq2}
u_0(t)=u_1(t)-h\alpha(t) & t >0,\\
\label{eq3}
u_{N+1}(t)=u_{N}(t) & t >0,\\
\label{eq4}
u_i(0)=v_i. & 
\end{eqnarray}
in which we have implemented Neumann boundary condition to calculate the values $u_{N+1}(t)$ and $u_{0}(t)$. 

As a time marching scheme, we will use a semi implicit one, in order to have a good stability: it is important since, as we will see hereafter, the calculation of optimal $\alpha(t)$ requires a great number of numerical solutions of this system and a not too small time step $\Delta t$ must be used.

We fix a value for the final time $T$, $T$ being large enough to obtain a steady state:  in practice the solution converges toward equilibrium relatively fastly for small values of $\epsilon$, which is the case here. The time interval $[0,T]$ is subdivided into $M$ subintervals of length $\Delta t$, $[t_k,t_{k+1}]$, $k=0,\cdots M-1$, with $t_k=k\Delta t$ and
$t_M=T$. We note $u^M$ the numerical approximation to ${\bar u}$.

A first idea is to compute $\alpha(t)$ as a piecewise constant function  in time, say
\begin{eqnarray}
\alpha(t)=\displaystyle{\sum_{k=0}^{M-1} \alpha_k\chi_{[t_k,t_{k+1}]}}.
\end{eqnarray}
\\

We remark here that calculating $\alpha(t)$ as a piecewise constant function, is a simple and stable approach: other techniques allowing orthogonal polynomial (such as Fourier or Laguerre) could be used but the strong decreasing of the Fourier coefficients
make the problem ill-conditioned. Also, an heuristic method can be used for accelerating the numerical convergence: 
once computed $\alpha^M(t)$ for a given $M$ and $\Delta t$, one can repeat the computation for $2M$ time steps (with $\Delta t/2$), in order to get a more precise result $\alpha^{2M}(t)$,  starting from a mid-point interpolated value.\\

The problem we want to solve is then expressed as
\\
\begin{equation}
\mbox{find $(\alpha_0, \cdots,\alpha_{M-1})$ such as minimizing $N(u^M)$}.
\label{opt_pb}
\end{equation}
\\
We can now describe the global approach under the following algorithm
\begin{center}
\begin{minipage}[H]{12cm}
  \begin{algorithm}[H]
    \caption{Optimal Configuration Search}\label{OptSearch}
    \begin{algorithmic}[1]
    \State Initialization: Start from an initial guess $\alpha^{(0)}(t)=\displaystyle{\sum_{k=0}^{M-1} \alpha^{(0)}_k\chi_{[t_k,t_{k+1}]}}$
     \For{$m=0,1, \cdots$until convergence}
     \State Compute $u^M$ by time integration of (\ref{eq1})-(\ref{eq4}) with 
     $$\alpha^{(m)}(t)=\displaystyle{\sum_{k=0}^{M-1} \alpha^{(m)}_k\chi_{[t_k,t_{k+1}]}}$$
     \State Compute $N(u^M)$.
     \State Update $(\alpha^{(m+1)}_0, \cdots,\alpha^{(m+1)}_{M-1})$ from $\alpha^{(m)}$ by a derivative free  optimization process  (nonlinear search)
            \EndFor
    \end{algorithmic}
    \end{algorithm}
\end{minipage}
\end{center}

Also, we will have to establish numerical convergence by varying $M$, $\Delta t$.
\subsection{Practical solution to the optimization problem}
\subsubsection{Full discretization scheme of the equations}
We subdivide $[0,T]$ into $M$ subintervals of length $\Delta t$ and note $u^k_i$ the approximation of $u_i(t_k)$ at time $t_k=k \Delta t$.
Let $A$ be the discretization matrix of the negative seconde derivative en $x$ with homogeneous Neumann boundary conditions and $U^k=(u_0^k,\cdots,u^k_{n+1})^T$.
 We consider the following linearized implicit Euler scheme which reads, after usual simplifications as
\begin{eqnarray}
\label{Scheme1}
U^{k+1}+\Delta t A U^{k+1}+\Frac{\Delta t}{\epsilon^2}(U^k)^2U^{k+1}=U^k+\Frac{\Delta t}{\epsilon^2}U^k +\Frac{\Delta t}{h}F^k,
\end{eqnarray}
where $F^k=(-\alpha(t_k),0\cdots,0)^T$.
We have the
\begin{proposition}
Assume $\alpha=0$. If $\Delta t < \epsilon^2$ and if $-1\le U^0\le 1$ then
the sequence $U^k$ defined by the scheme (\ref{Scheme1}) satisfies $-1\le U^k\le 1, \forall k$.
\end{proposition}
\begin{proof}
If $\alpha^k=0$, we set $\overline{U^k}=U^k-1$ and $\underline{U^k}=U^k+1$. We'll show by induction that
$\overline{U^k}\le 0$ and $\underline{U^k}\ge 0 \forall k \in \N$. Let us fix $k$ and assume that $-1\le U^k\le 1$ says $\overline{U^k}\le 0$ and $\underline{U^k}\ge 0$.
We have, after simplifications
$$
\left(1+\Frac{\Delta t}{\epsilon^2} (\overline{U^k}+1)^2\right)\overline{U^{k+1}}
+\Delta t A \overline{U^{k+1}}=\left(1-\Frac{\Delta t}{\epsilon^2}\right)\overline{U^k} -\Frac{\Delta t}{\epsilon^2}(\overline{U^k})^2.
$$
If $\overline{U^k}\le 0$, then $\left(1+\Frac{\Delta t}{\epsilon^2} (1+\overline{U^k}^2+\overline{U^k}) \right) \ge 1-\Frac{\Delta t}{\epsilon^2}>0$, so the matrix
$$
M^k=diag\left(1+\Frac{\Delta t}{\epsilon^2} (1+\overline{U^k}) ^2\right)
+\Delta t A
$$
has the discrete maximum property ($M_k$ is a M-matrix) and since $\overline{U^k}\le 0$ and $\Frac{\Delta t}{\epsilon^2}<1$, we have $\left(1-\Frac{\Delta t}{\epsilon^2}\right)\overline{U^k} -\Frac{\Delta t}{\epsilon^2}(\overline{U^k})^2 \le 0$ so that
$$
\overline{U^{k+1}}\le 0.
$$
We proceed in a similar way for $\underline{U^k}$.
\end{proof}

As, we will see hereafter in the numerical simulations, this is observed for moderate values of $\alpha$.\\

It is important to underline that it is crucial that the numerical scheme heritates of the intrinsic properties of the equation (maximum principle, asymptotic behavior) for calculating a numerical control: if not, non-physical control can be determined\cite{Ignat1}.
\subsubsection{Choice of the optimization method}
As pointed out before, the principal key of this approach is the choice of the optimization method to compute $\alpha(t)$. This is not an easy task since the merit function is not differentiable in the $L^1$ case; generally the gradient is not available as in the $L^2$ case. Hence, gradient methods cannot be used. We have then to address to derivative free optimization algorithms \cite{ConnScheinbergVicente}.\\

In order to illutrate the numerical cacutation of optimal $\alpha$, we use Matlab build-in house solver, such as
{\tt x = fminsearch(fun,x0)} or {\tt x = fminunc(fun,x0)} where {\tt fun} is the objective function which in our case is built numerically by solving the Allen--Cahn equation up to $t=T$ with $\alpha$ as an input.

\section{Numerical Results}
The minimization procedure used is {\tt fminunc} available in from the optimization toolbox of Matlab computational software, \cite{Matlab} .
We fix $\epsilon$ and we start from $\alpha^0(t)=0, \forall t \in [0,T]$. We hereafter display
results for different values of $\epsilon$ and $T$.  The merit function is $N(u^M)$.
\subsection{Regular fixed initial data - numerical convergence}
Starting from a regular initial data allows to observe numerical convergence of the optimal control function $\alpha(t)$
as the number of discretisation points $N$ increases, and as, for a fixed $T$, the time step $\Delta t $ decreases.
We can see, in Figures (\ref{fig4})-(\ref{fig5})-(\ref{fig6})-(\ref{fig7}) the good coherence of the results for fixed values of $\epsilon=0.01$ and $T$ when varying $\Delta t$ and the number $N$ of grid points. In all the cases, the global procedure
allows minimizing the number of interphases or stripes.
\begin{figure}[h]
\begin{center}
\includegraphics[height=7.0cm,width=8cm]{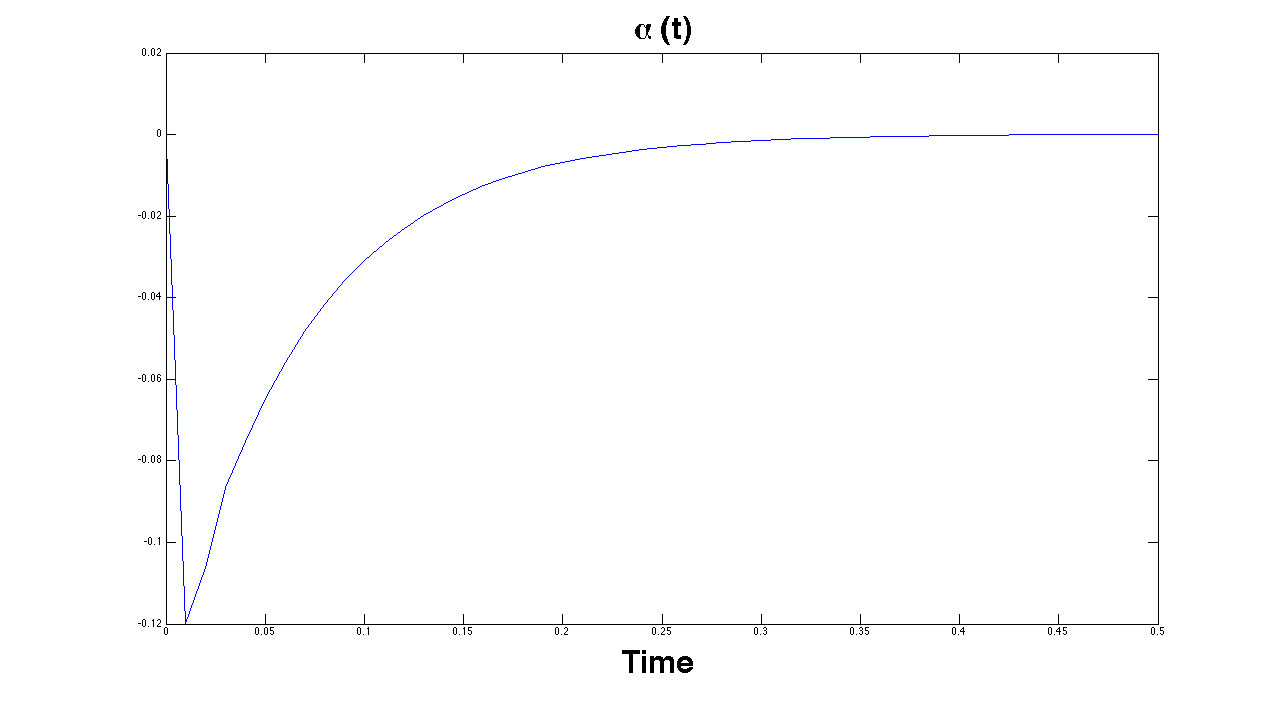}
\includegraphics[height=7.0cm,width=8cm]{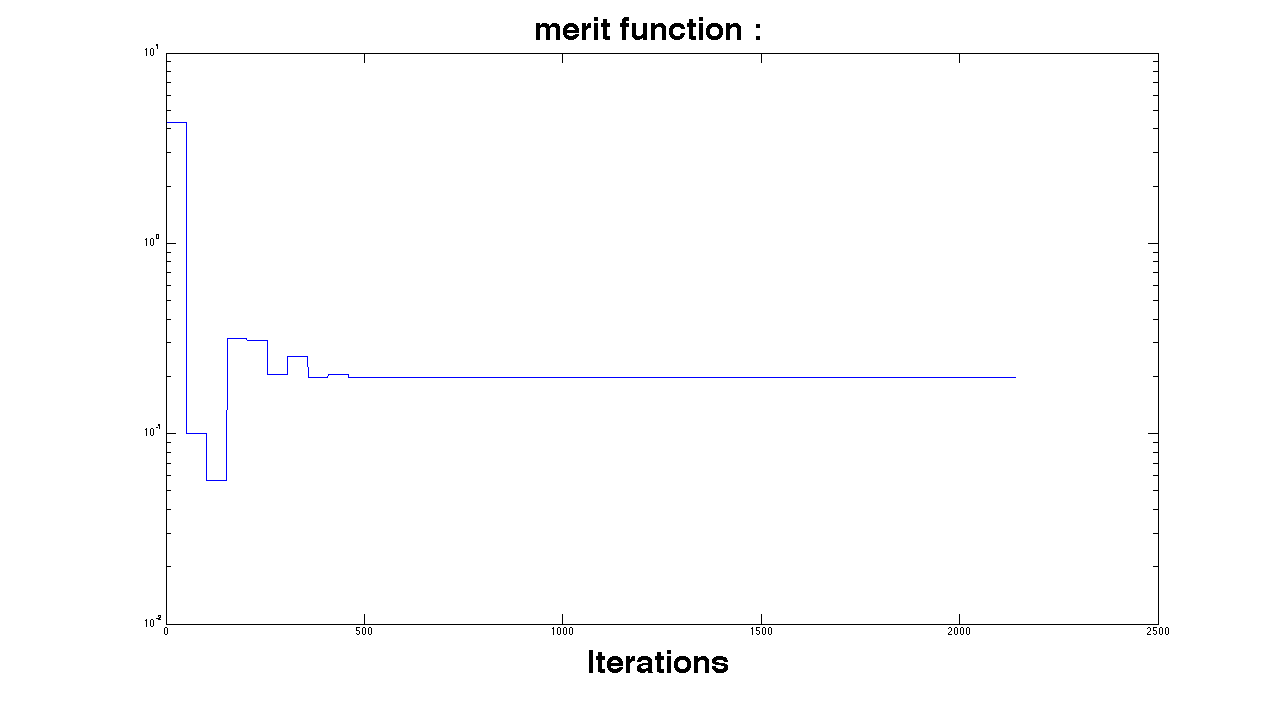}\\
\includegraphics[height=7.0cm,width=8cm]{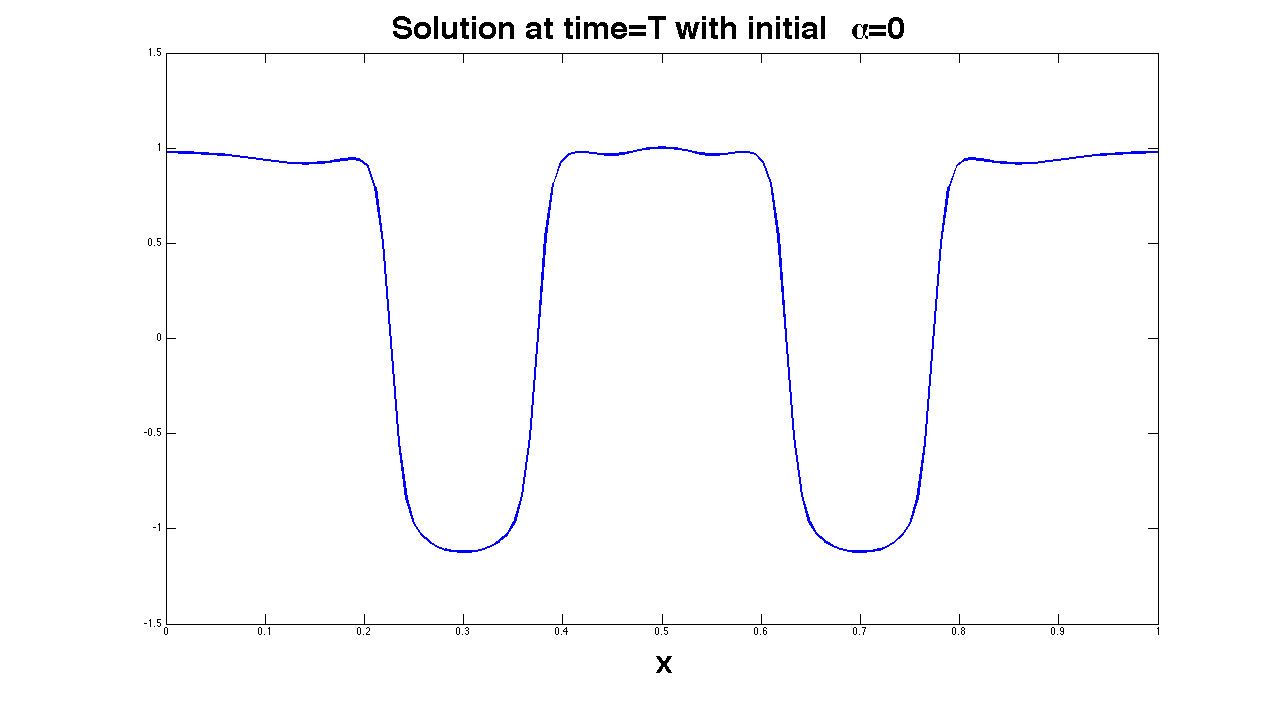}
\includegraphics[height=7.0cm,width=8cm]{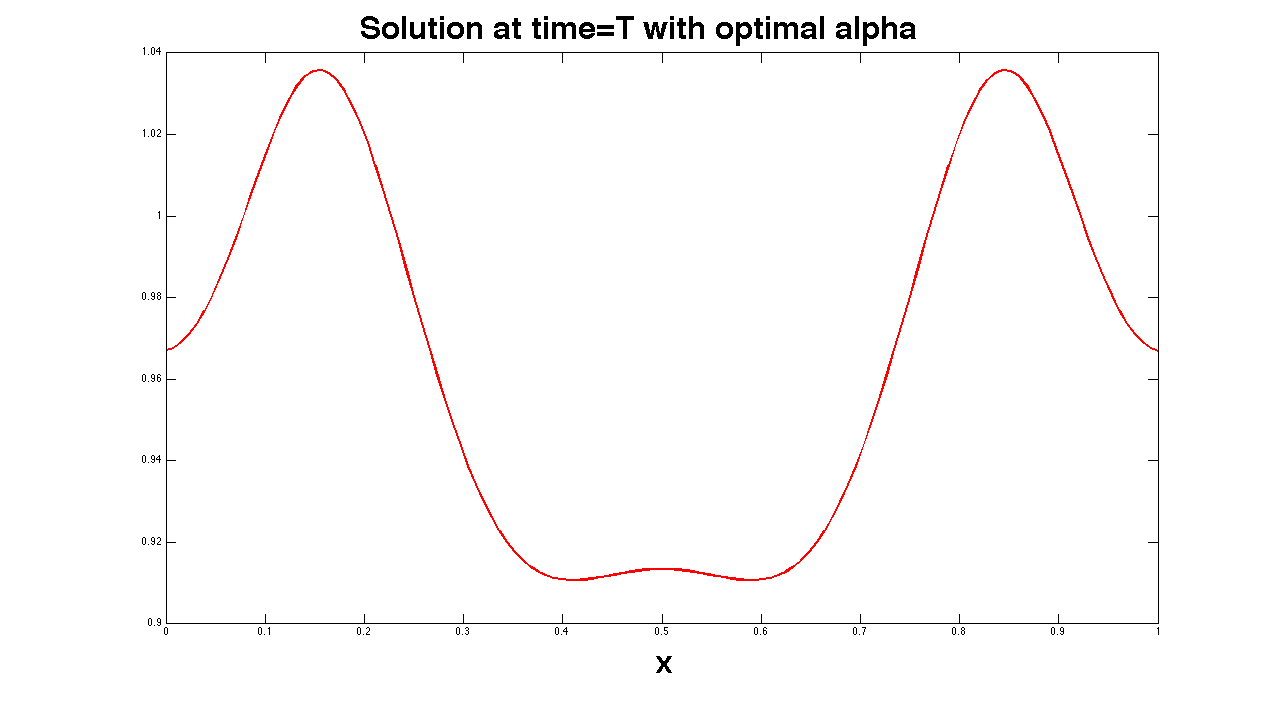}
\end{center}
\caption{$\epsilon=0.01, \ \Delta t = 0.01, \ T=0.5, u_0=\cos(20\pi x), \ N=127$}
\label{fig4}
\end{figure}
\clearpage
\begin{figure}[h]
\begin{center}
\includegraphics[height=7.0cm,width=8cm]{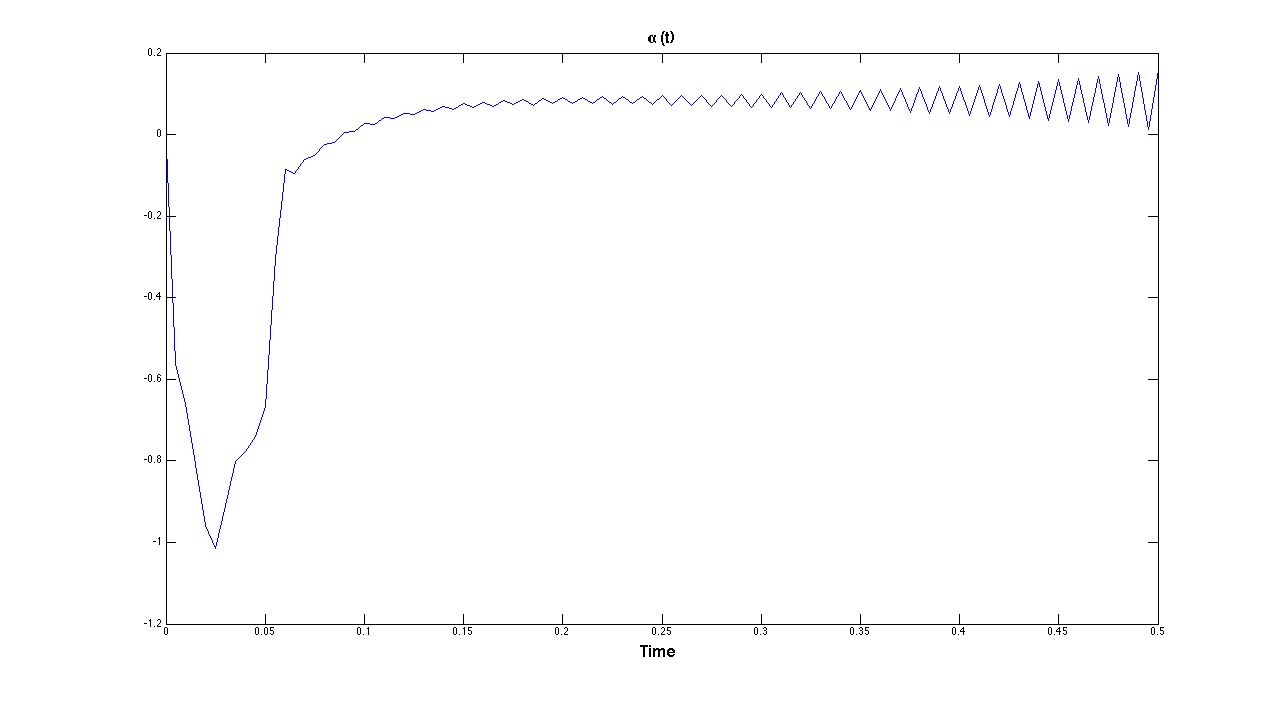}
\includegraphics[height=7.0cm,width=8cm]{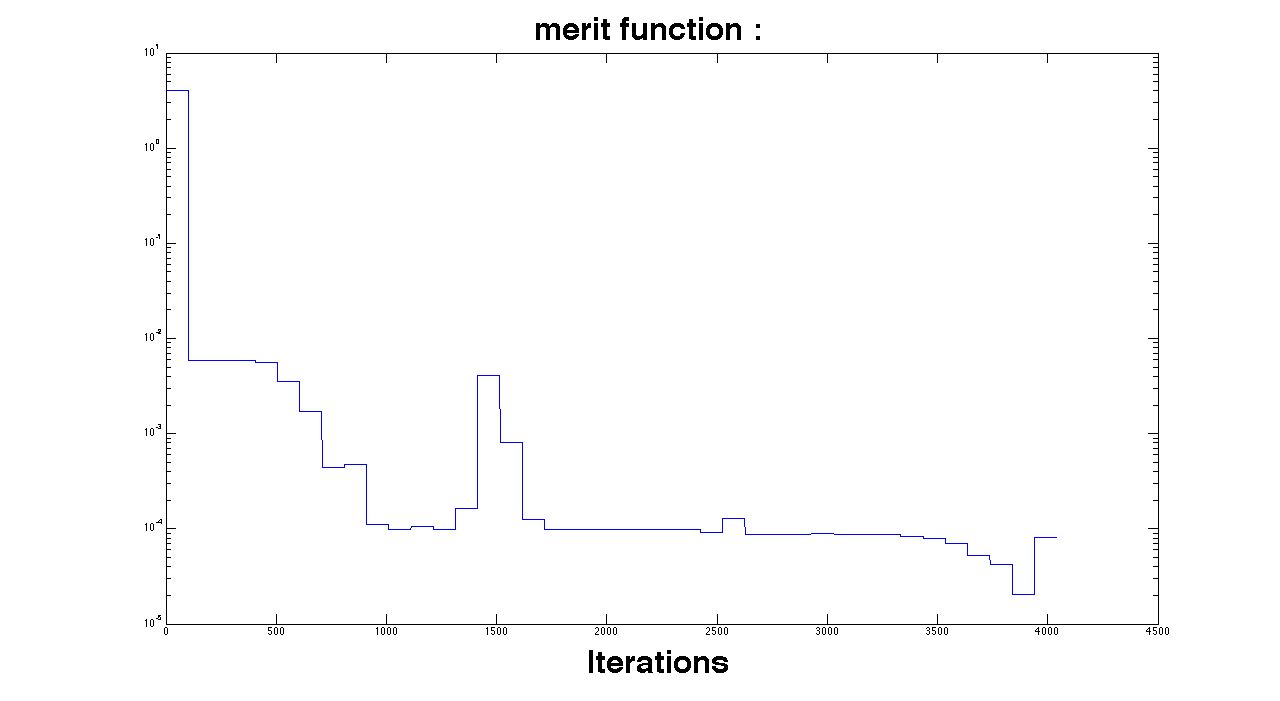}\\
\includegraphics[height=7.0cm,width=8cm]{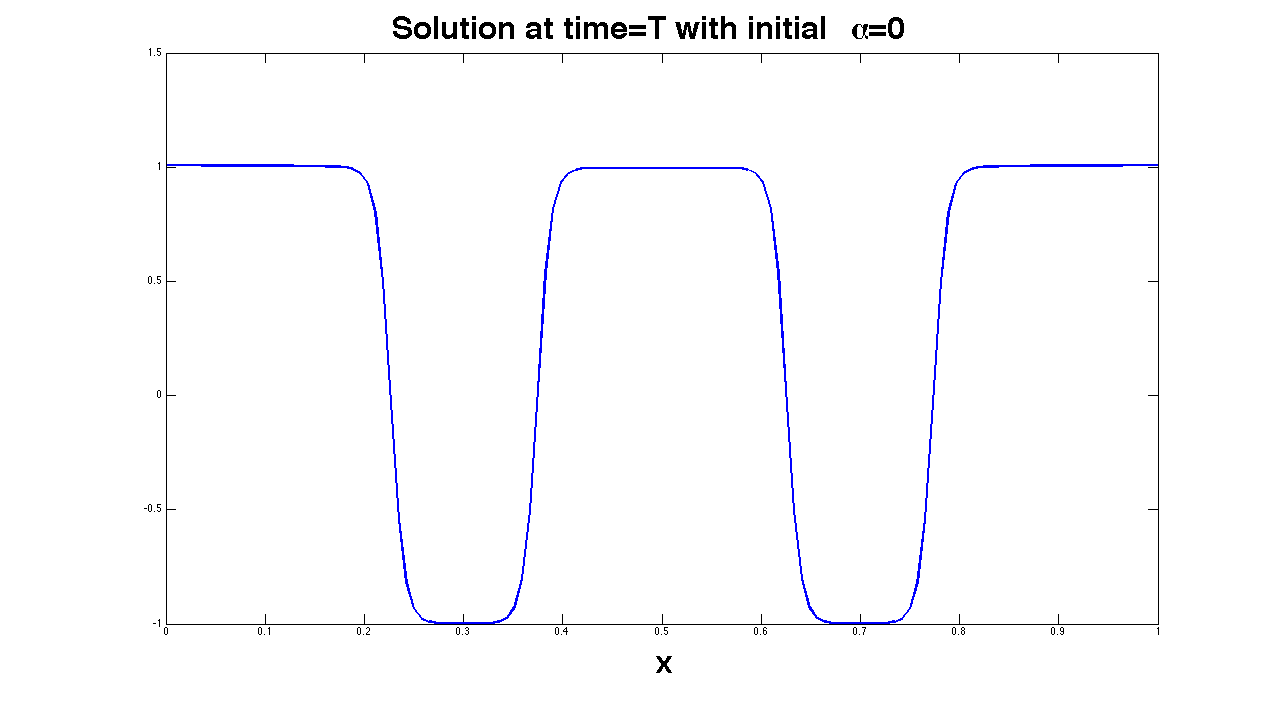}
\includegraphics[height=7.0cm,width=8cm]{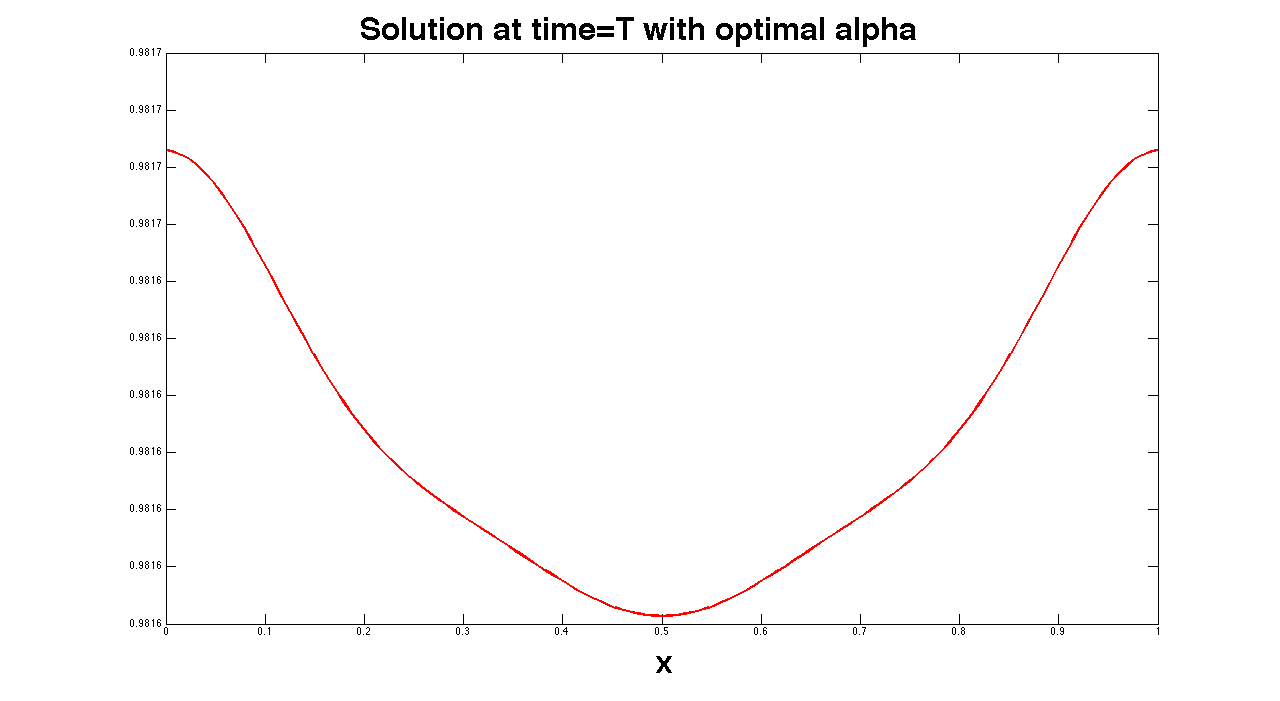}
\end{center}
\caption{$\epsilon=0.01, \ \Delta t = 0.005, \ T=0.5, u_0=\cos(20\pi x), \ N=127$}
\label{fig5}
\end{figure}
\clearpage
\begin{figure}[h]
\begin{center}
\includegraphics[height=7.0cm,width=8cm]{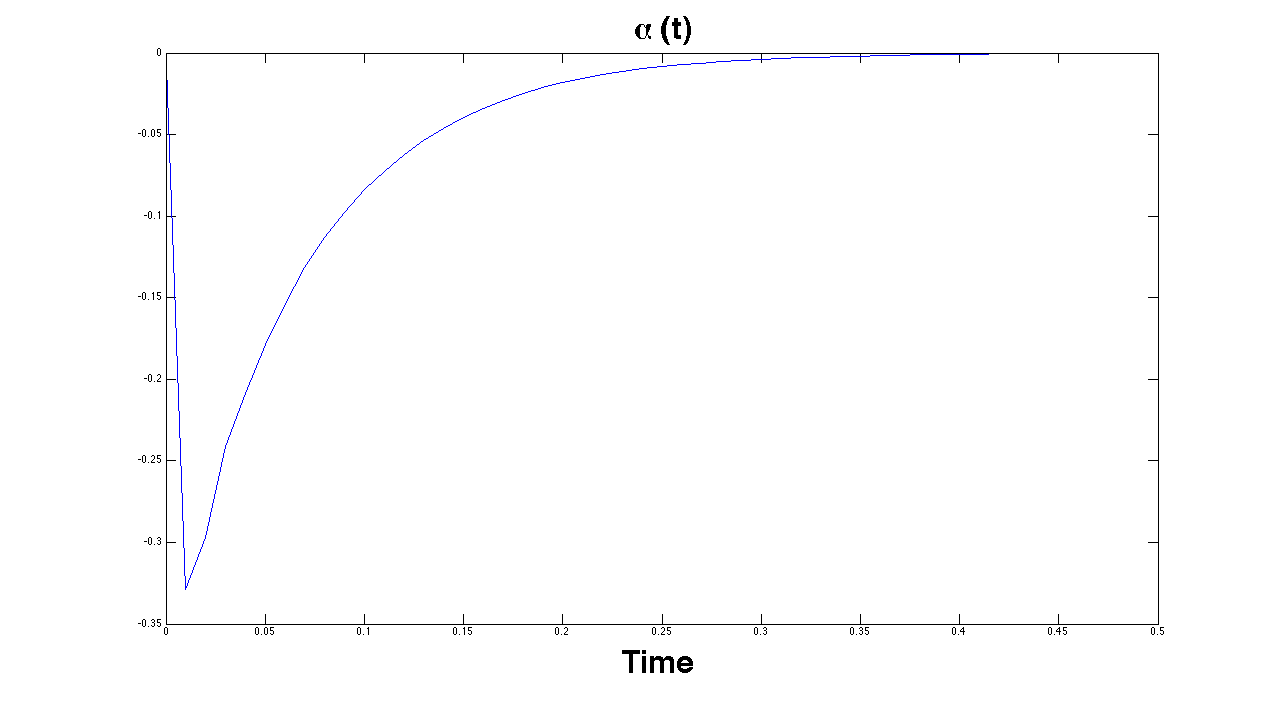}
\includegraphics[height=7.0cm,width=8cm]{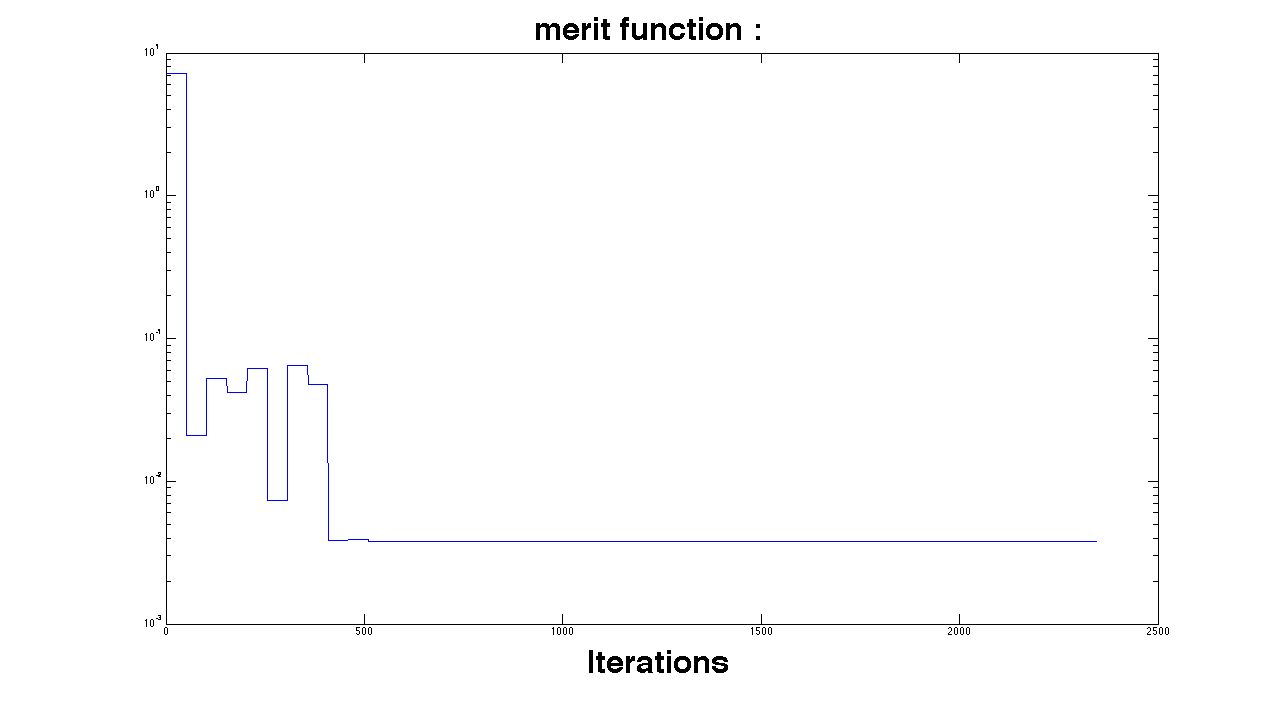}\\
\includegraphics[height=7.0cm,width=8cm]{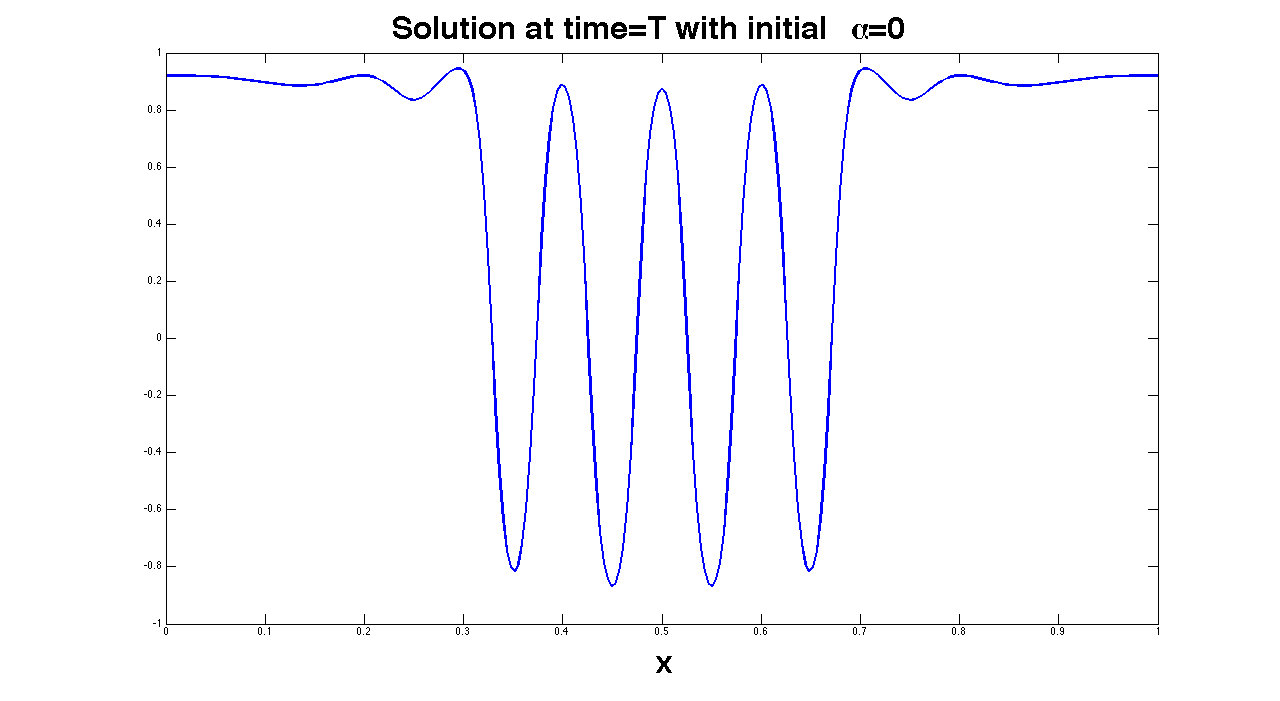}
\includegraphics[height=7.0cm,width=8cm]{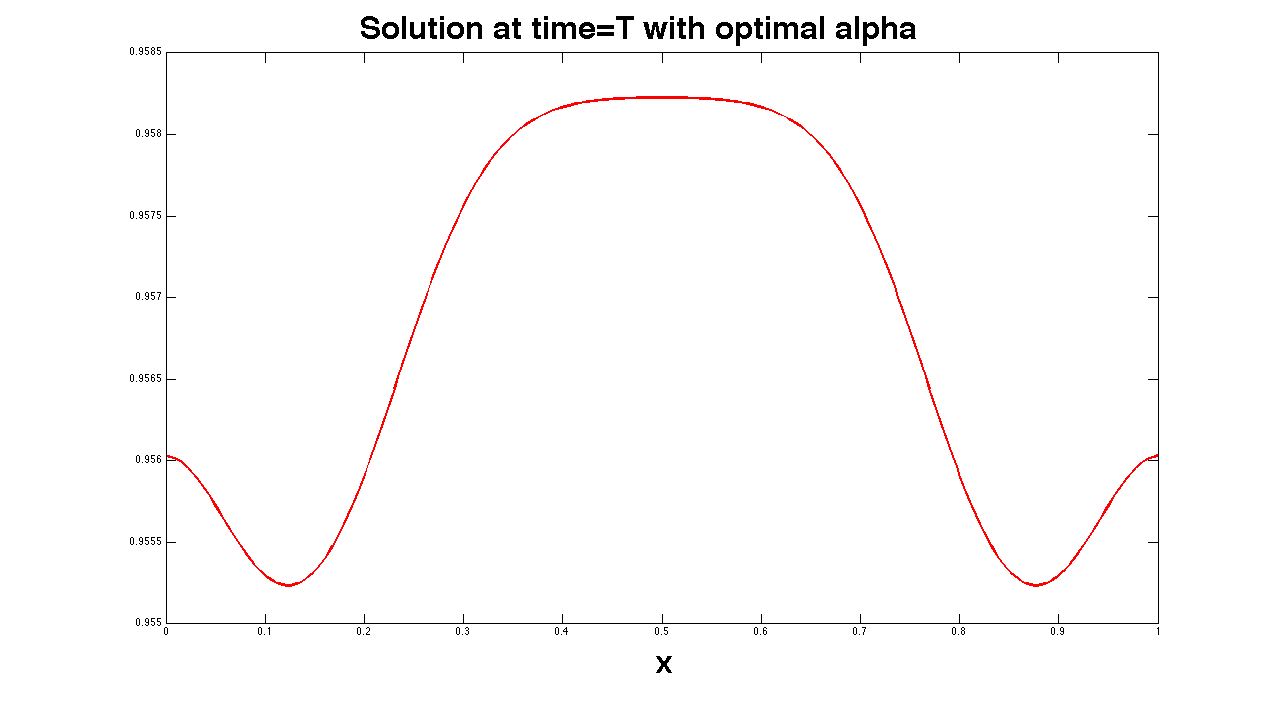}
\end{center}
\caption{$\epsilon=0.01, \ \Delta t = 0.01, \ T=0.5, u_0=\cos(20\pi x), \ N=255$}
\label{fig6}
\end{figure}
\clearpage
\begin{figure}[h]
\begin{center}
\includegraphics[height=7.0cm,width=8cm]{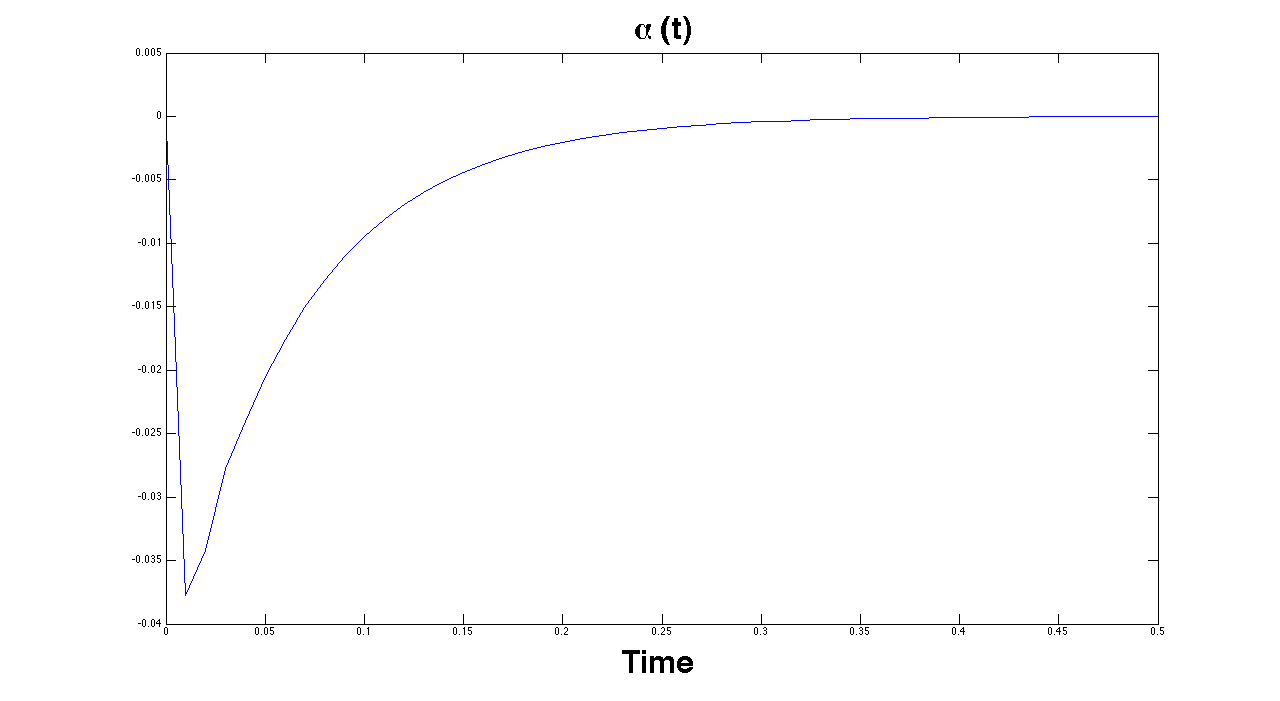}
\includegraphics[height=7.0cm,width=8cm]{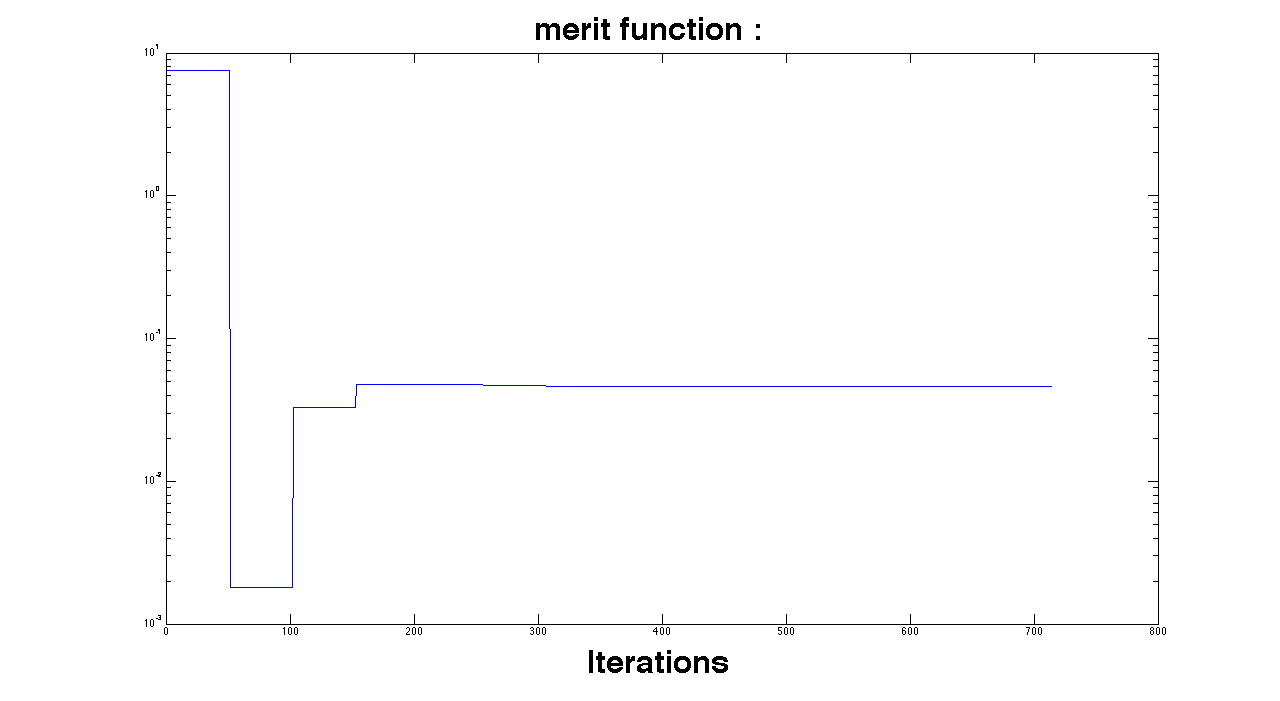}\\
\includegraphics[height=7.0cm,width=8cm]{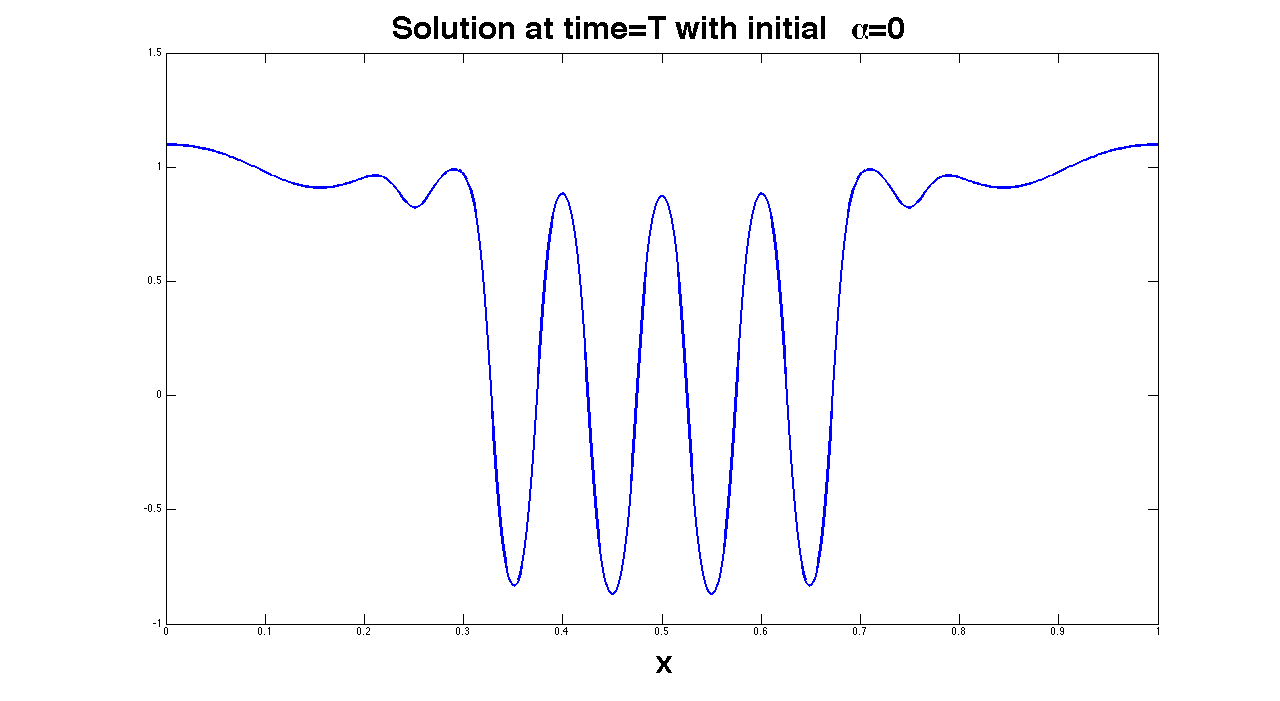}
\includegraphics[height=7.0cm,width=8cm]{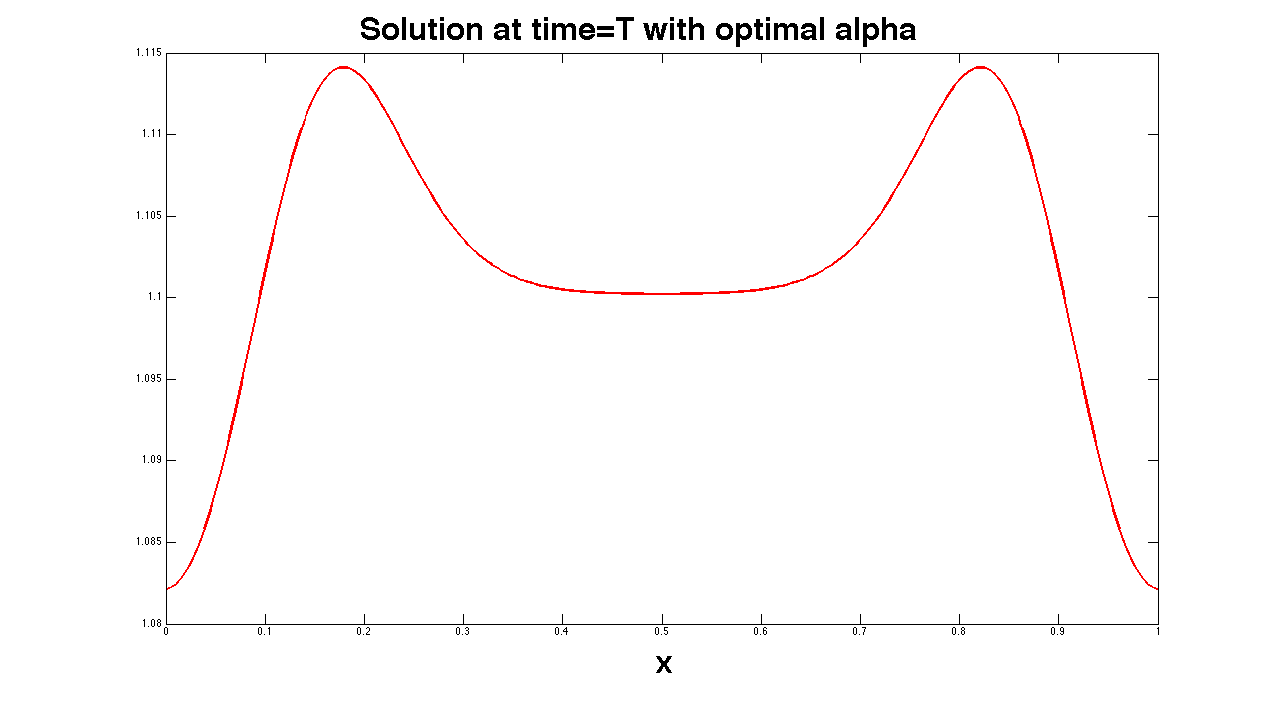}
\end{center}
\caption{$\epsilon=0.01, \ \Delta t = 0.01, \ T=0.5, u_0=\cos(20\pi x), \ N=511$}
\label{fig7}
\end{figure}
\clearpage
\subsection{Random initial data}
\subsubsection{Minimization of the number of interphases}
Here, in Figures (\ref{fig8}) and (\ref{fig9}) the discrete components of $u_0$ are randomly calculated following a uniform law on $[-1,1]$ and we start from $\alpha^0(t)=0, \forall t \in [0,T]$.  As we see, here again, in all cases, the global procedure, illustrating the effective numerical controllability by the boundaries and the robustness of the approach, since the data are very oscillating.
\begin{figure}[h]
\begin{center}
\includegraphics[height=7.0cm,width=8cm]{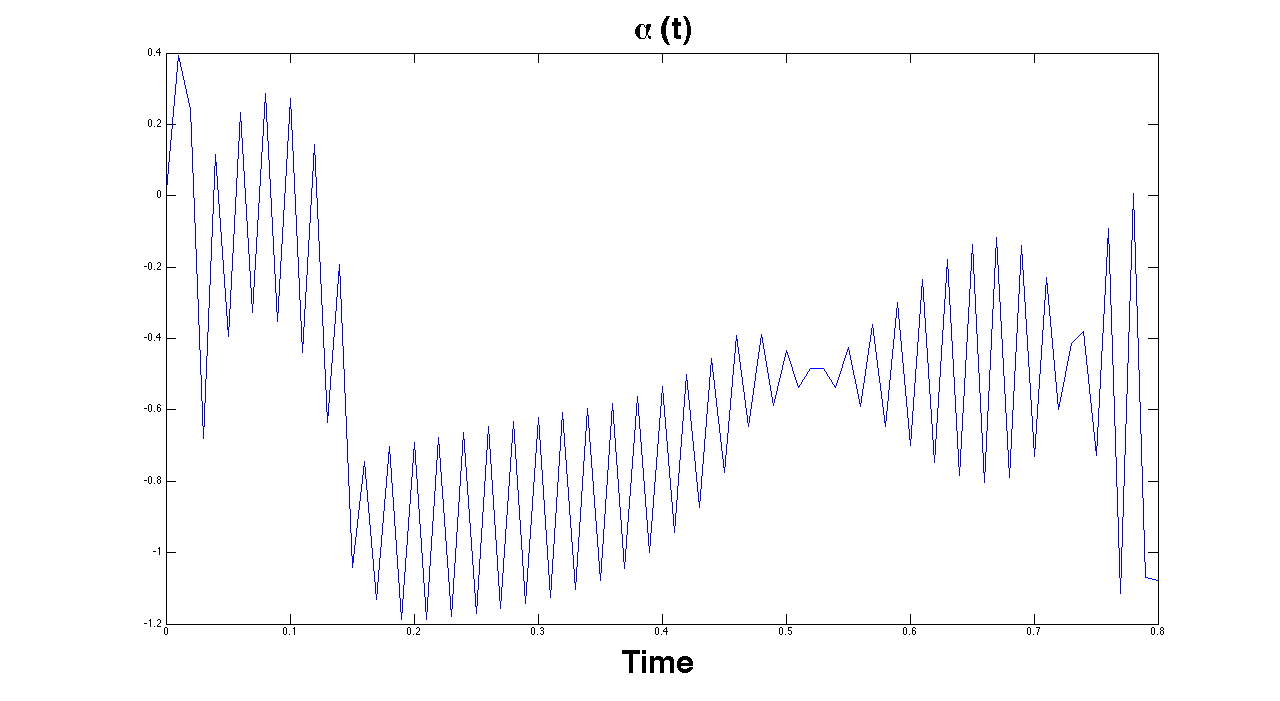}
\includegraphics[height=7.0cm,width=8cm]{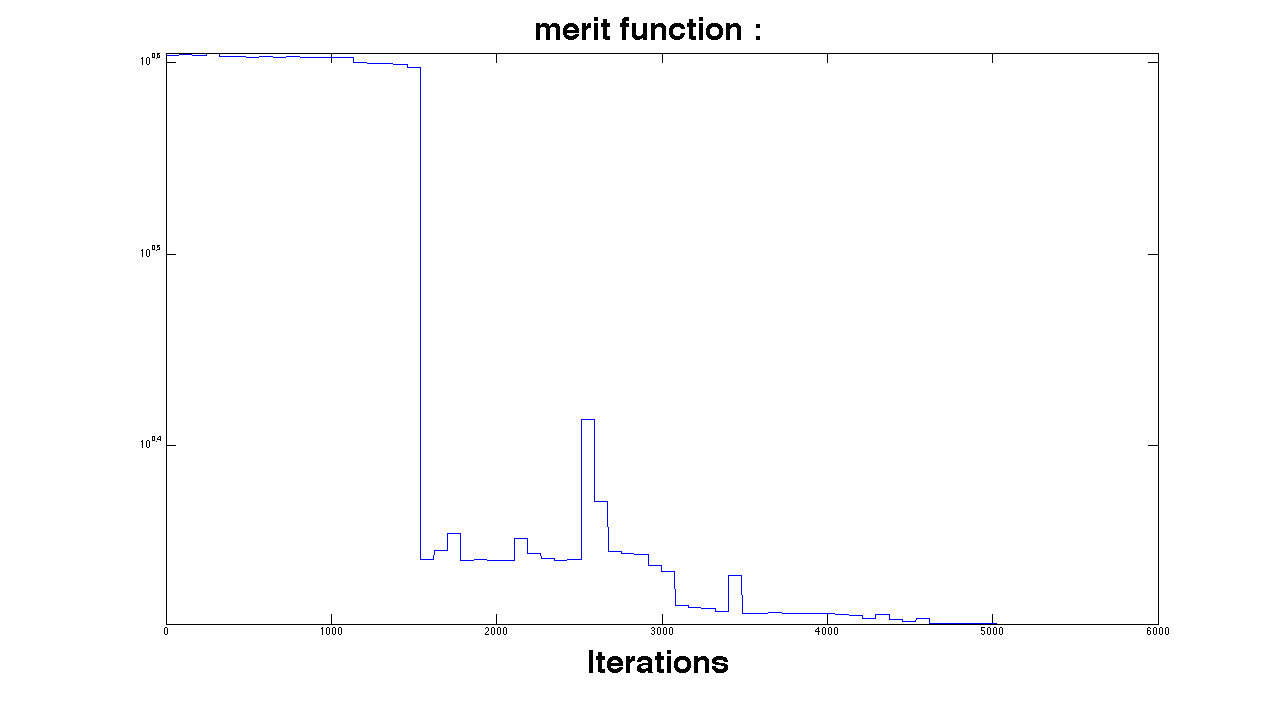}\\
\includegraphics[height=7.0cm,width=8cm]{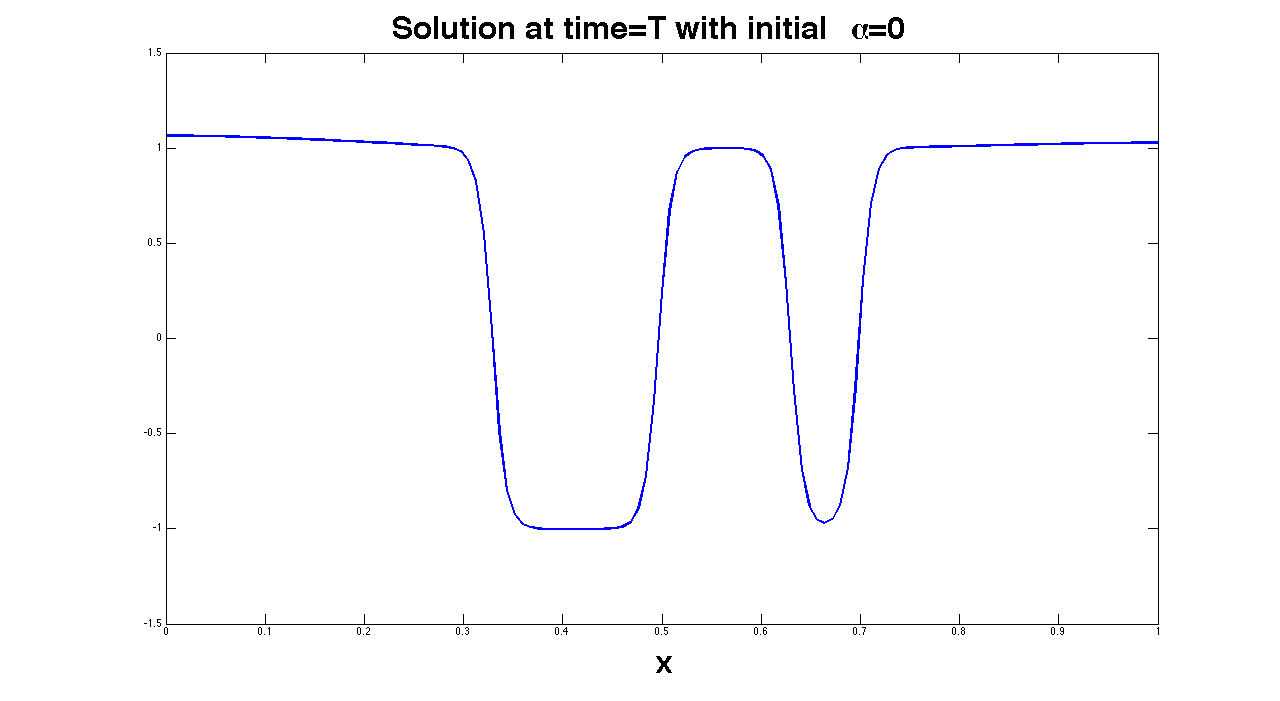}
\includegraphics[height=7.0cm,width=8cm]{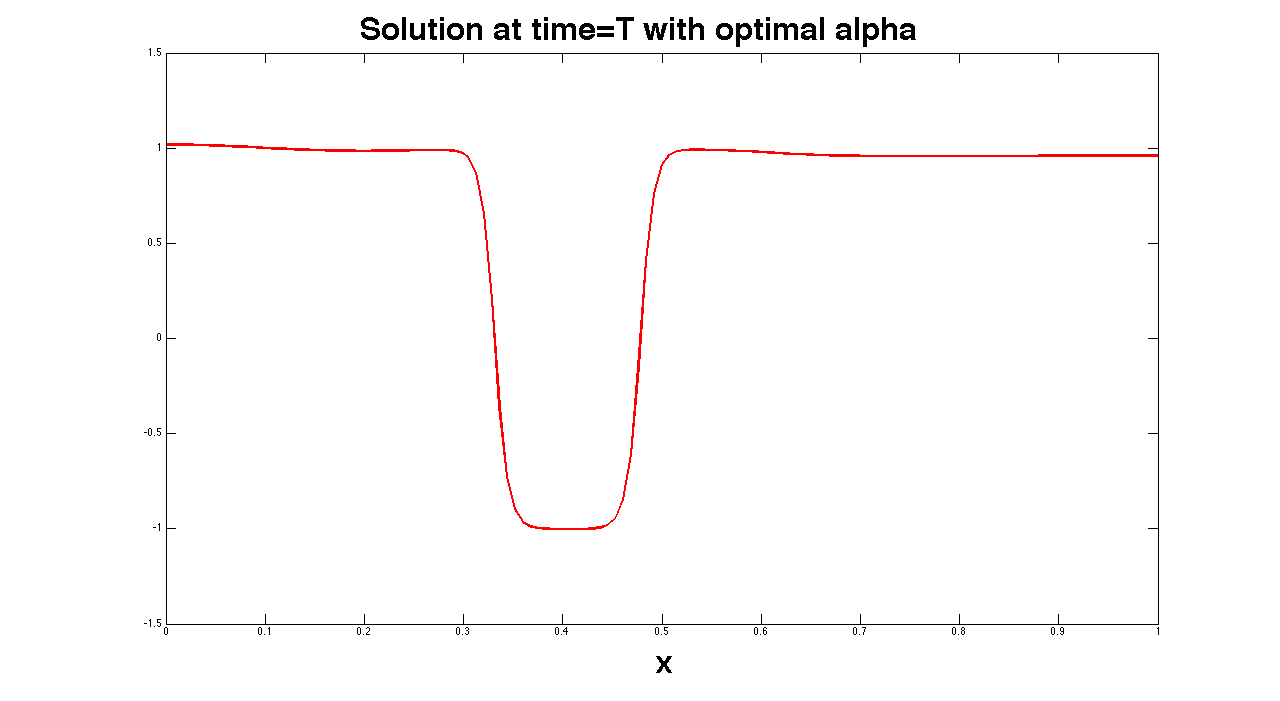}
\end{center}
\caption{$\epsilon=0.01, \ \Delta t = 0.01, \ T=0.8, u_0=rand, \ N=127$}
\label{fig8}
\end{figure}
\clearpage
\begin{figure}[h]
\begin{center}
\includegraphics[height=7.0cm,width=8cm]{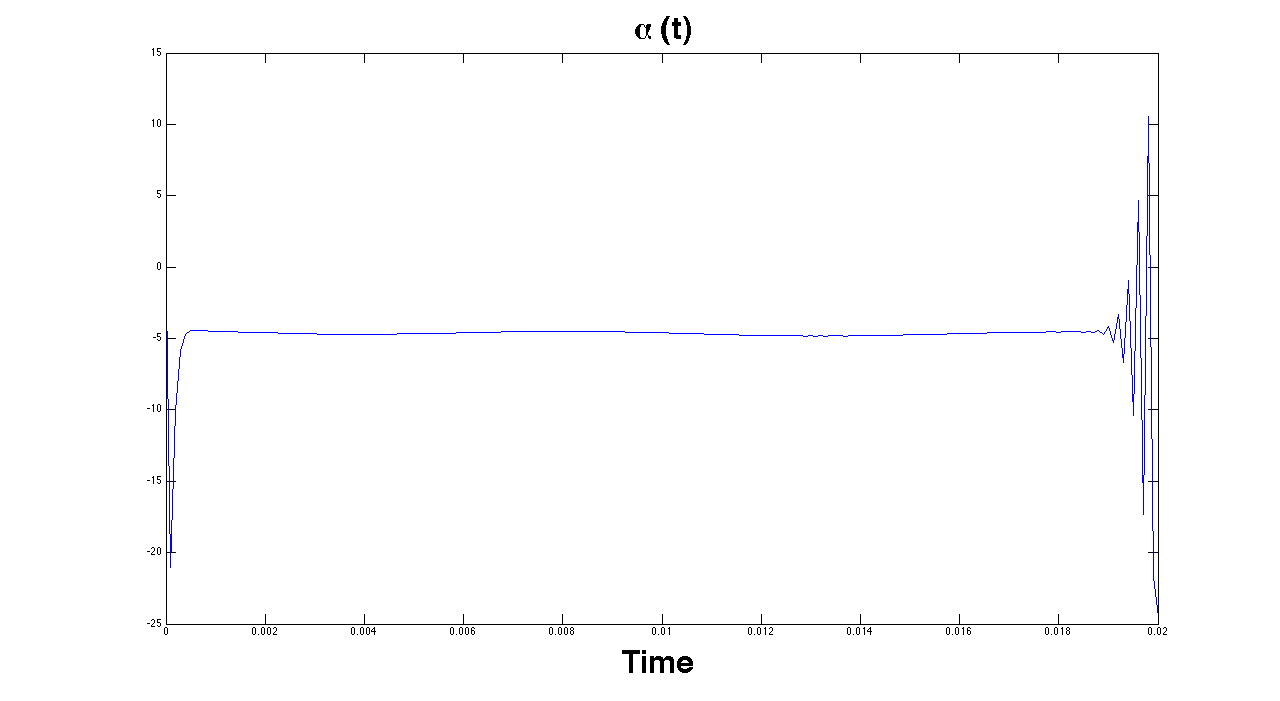}
\includegraphics[height=7.0cm,width=8cm]{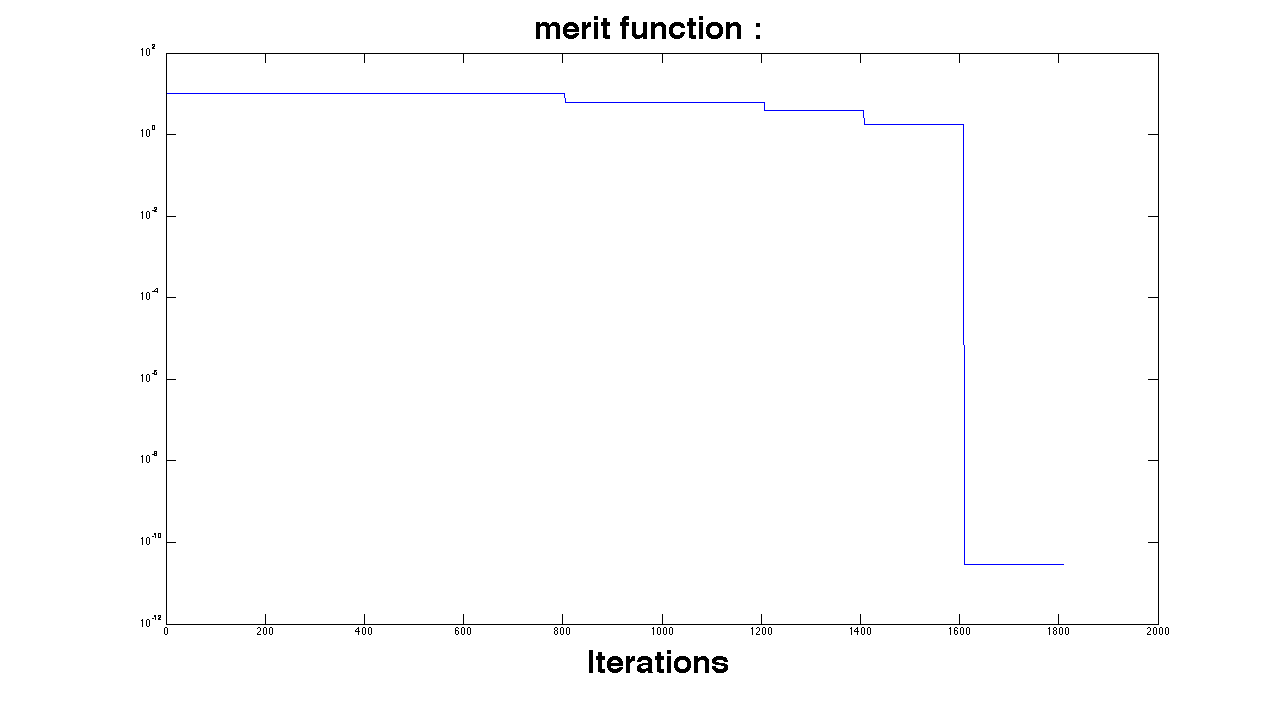}\\
\includegraphics[height=7.0cm,width=8cm]{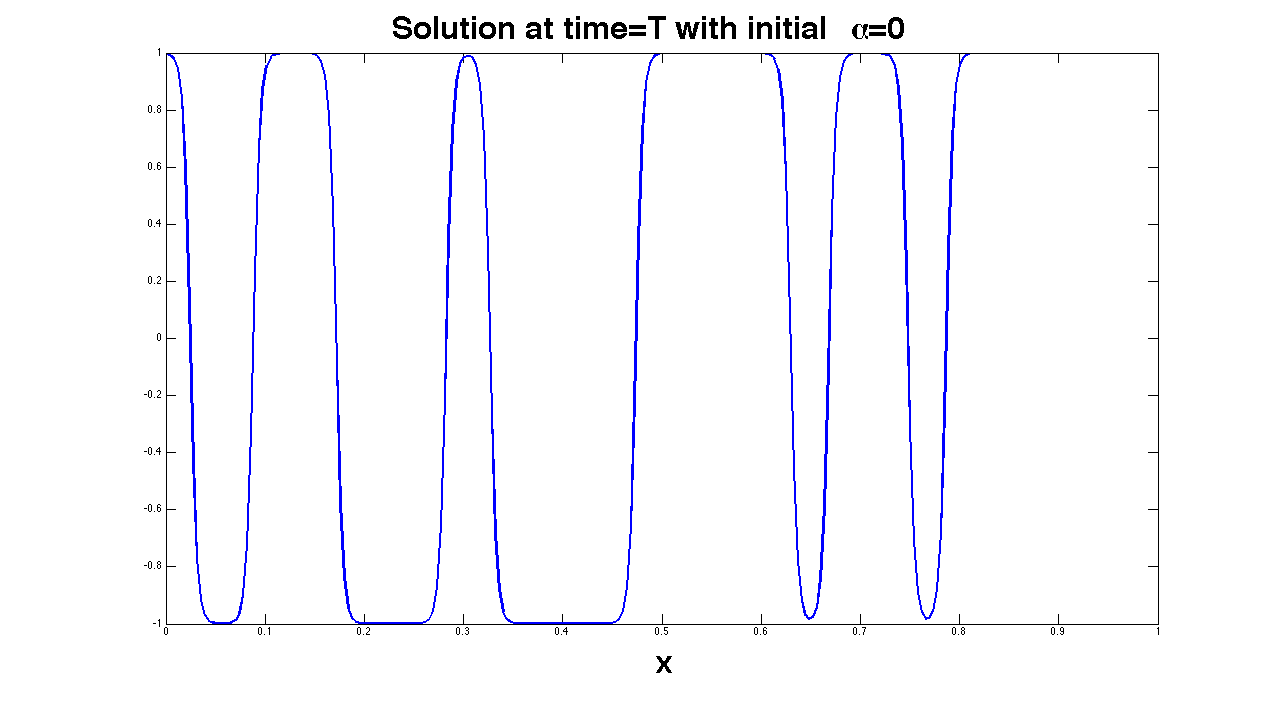}
\includegraphics[height=7.0cm,width=8cm]{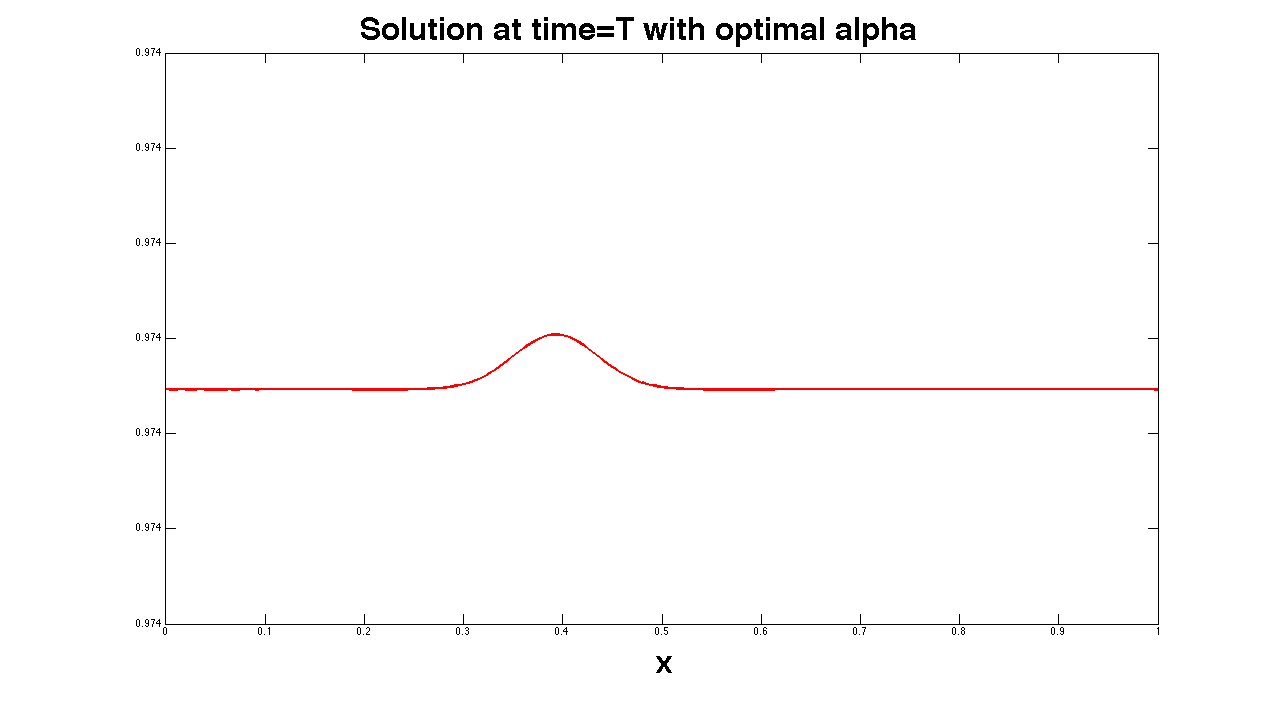}
\end{center}
\caption{$\epsilon=0.005, \ \Delta t = 0.0001, \ T=0.01, u_0=rand, \ N=255$}
\label{fig9}
\end{figure}
\clearpage
\subsubsection{Selection of a given phase}
Finally, we give hereafter a numerical illustration of the optimization process when considering
a weighted merit function as in (\ref{weighted}): we adopt here
$F(u)=10 \parallel u-1\parallel +\parallel u+1\parallel $. As we see in Figure (\ref{fig10}), the merit function favors the formation of
the phase $U=1$, the initial datum is, as above, randomly generated on $[-1,1]$ by an uniform law.
\begin{figure}[h]
\begin{center}
\includegraphics[height=7.0cm,width=8cm]{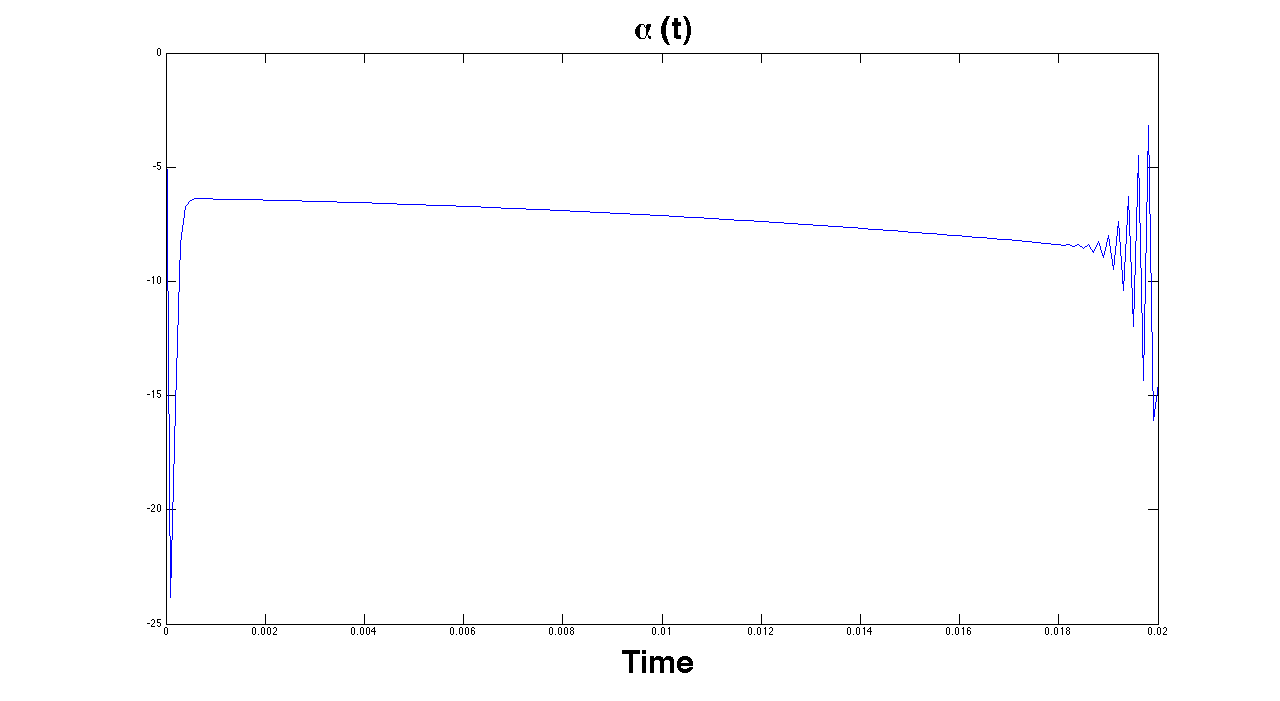}
\includegraphics[height=7.0cm,width=8cm]{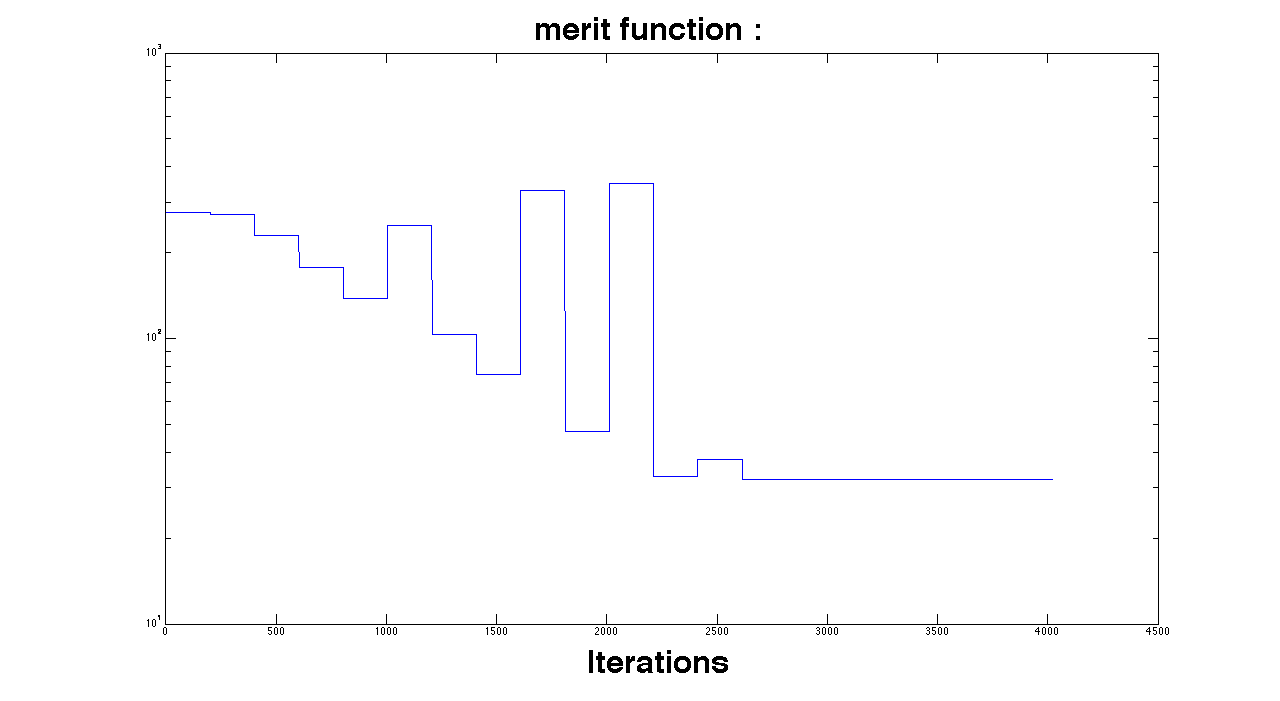}\\
\includegraphics[height=7.0cm,width=8cm]{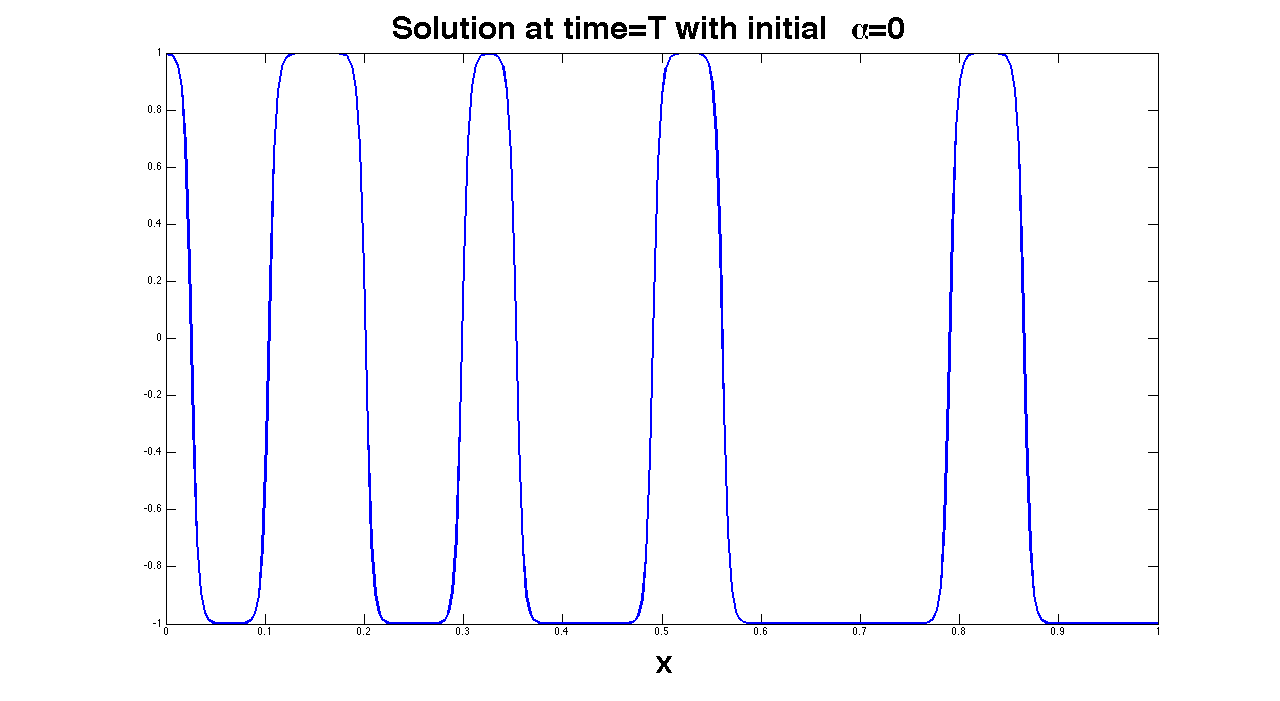}
\includegraphics[height=7.0cm,width=8cm]{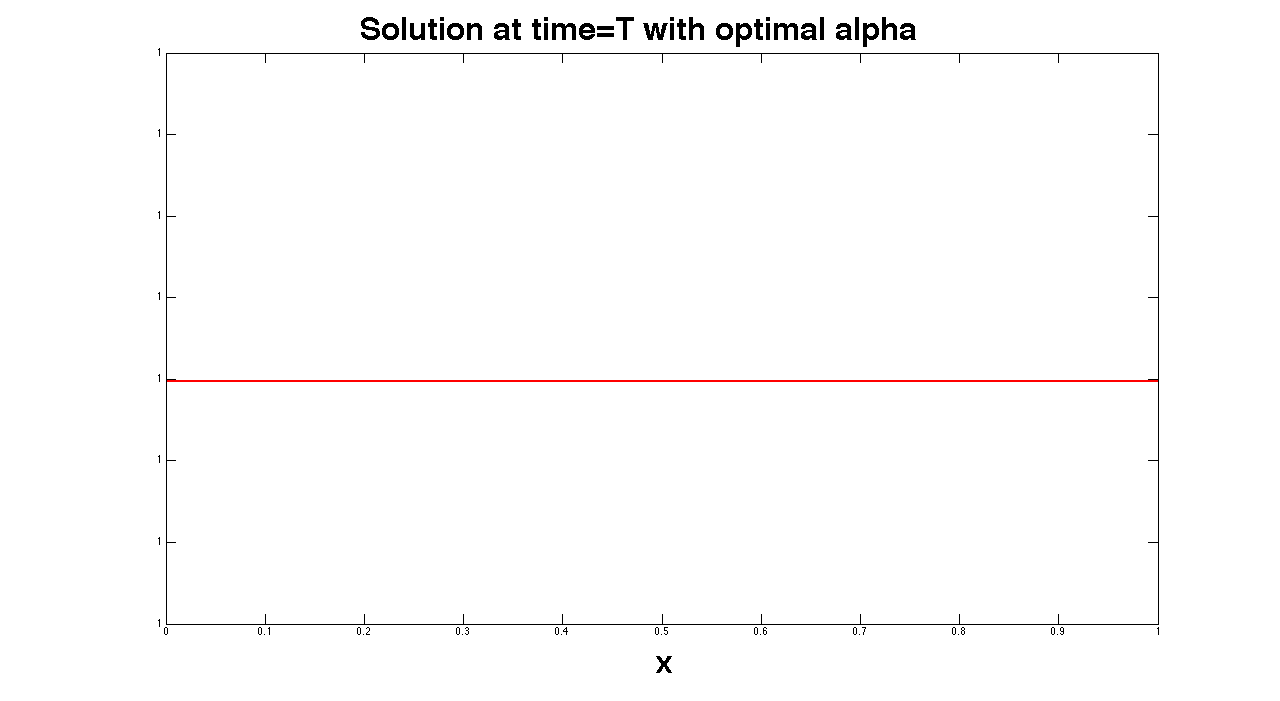}
\end{center}
\caption{$\epsilon=0.005, \ \Delta t = 0.0001, \ T=0.02, u_0=rand, \ N=255$}
\label{fig10}
\end{figure}
\section{Concluding remarks and features}
In this paper, we have presented a simple approach to calculate numerically a boundary control that allows obtaining an optimal steady-sate configuration, i.e., with a minimal number of interphases. We have also demonstrated that we can also favor the formation of a given phase by following the same procedure. The results we obtained are encouraging and show the numerical faisability of the proposed method.. 
Of course, we have here considered first a relatively simple case, namely the one dimensional case  before  extending the approach to 2D or 3D models which correspond to more realistic situations found, for example, in electrochemistry. Furthermore, the monitoring of the number of interphases by $\epsilon$ (problem 2) is an important feature  that we will study in a near future. Finally the integration of such optimization algorithms in an in-house multiscale simulator of electrochemical power generators will be also considered \cite{FrancoLiberT}.

\end{document}